\documentclass[11pt, a4paper]{article} 

\usepackage{a4}
\usepackage{geometry}

\usepackage{latexsym}
\usepackage{amssymb}
\usepackage{amsfonts}
\usepackage{mathrsfs}
\usepackage{amsmath}
\usepackage{amsthm}
\usepackage[colorlinks=true,linkcolor=black,citecolor=black]{hyperref}
\usepackage[pdftex]{graphicx}

\newcommand{\bfi}{\begin{fig}}
\newcommand{\efi}{\end{fig}}

\newcommand{\btab}{\begin{tab}}
\newcommand{\etab}{\end{tab}}

\newcommand{\barr}{\begin{array}}
\newcommand{\earr}{\end{array}}

\newcommand{\beqq}{\begin{equation}}
\newcommand{\eeqq}{\end{equation}}

\newcommand{\beao}{\begin{eqnarray*}}
\newcommand{\eeao}{\end{eqnarray*}\noindent}

\newcommand{\beam}{\begin{eqnarray}}
\newcommand{\eeam}{\end{eqnarray}\noindent}

\newcommand{\bdis}{\begin{displaymath}}
\newcommand{\edis}{\end{displaymath}\noindent}

\def\bbr{{\Bbb R}}   
\def\bbn{{\Bbb N}}

\newcommand{\eps}{{\varepsilon}}

\newcommand{\vp}{{\varphi}}

\newcommand{\ov}{\overline}
\newcommand{\un}{\underline}
\newcommand{\wh}{\widehat}
\newcommand{\wt}{\widetilde}

\newcommand{\mal}{\stackrel{\mbox{\tiny$\bullet$}}{}}

\newtheorem{Satz}{Theorem}[section]

\newtheorem{Proposition}[Satz]{Proposition}

\newtheorem{Definition}[Satz]{Definition}
\newtheorem{Beispiel}[Satz]{Example}
\newtheorem{Bemerkung}[Satz]{Remark}

\newtheorem{Notiz}[Satz]{Note}

\usepackage{stmaryrd}
\usepackage{amsmath,amsthm}
\usepackage{amssymb}
\usepackage[numbers,square]{natbib}

\newtheorem{thmS}[Satz]{Theorem}

\newtheorem{thmBSP}[Satz]{Example}

\usepackage{graphicx}
\usepackage{subfigure}
\usepackage{framed}

\begin{document}
\bibliographystyle{apalike}
\numberwithin{equation}{section}

\title{How local in time is the no-arbitrage property under capital gains 
taxes~?}
\author{Christoph K\"uhn\thanks{Institut f\"ur Mathematik, Goethe-Universit\"at Frankfurt, D-60054 Frankfurt a.M., Germany, e-mail: ckuehn@math.uni-frankfurt.de\newline I would like to thank the editor,  Prof.~Riedel, and an anonymous associate editor for their valuable comments. I am especially grateful to the anonymous referee for finding a minor error in the previous version of Proposition~\ref{19.5.2017.1} and for many valuable suggestions that lead to a substantial improvement of the 	presentation of the results.}}
\date{}
\maketitle
\sloppy

\begin{abstract}
In frictionless financial markets, no-arbitrage is a local property in time. This means that a 
discrete time model is arbitrage-free if and only if there does not exist a one-period-arbitrage. With capital gains taxes, this equivalence fails. For a model with a linear tax and one non-shortable risky stock, we introduce the concept of
{\em robust local no-arbitrage} (RLNA) as the weakest local condition which guarantees dynamic no-arbitrage. Under a sharp dichotomy condition, we prove
(RLNA). 
Since no-one-period-arbitrage is necessary for no-arbitrage, the latter is sandwiched between  two local conditions, which allows us to estimate its non-locality. 

Furthermore, we construct a stock price process such that two long positions in the same stock hedge each other. This puzzling phenomenon that cannot occur in arbitrage-free frictionless markets (or markets with proportional transaction costs) is used to show that no-arbitrage alone does not imply the existence 
of an equivalent separating measure if the probability space is infinite. 

Finally, we show that the model with a linear tax on capital gains can be written as a model with proportional
transaction costs by introducing several 
fictitious securities.    
\end{abstract}

\begin{tabbing}

 
\end{tabbing}
\section{Introduction}
\setcounter{equation}{0}

In most countries, trading gains have to be taxed, which constitutes a major market friction. Tax systems are usually realization based, i.e., 
gains on assets are taxed when assets are sold and not when gains actually accrue. 
Consequently, investors hold various tax timing
options. Especially in the case of positive interest rates and
direct tax credits for 
losses, there is an incentive to realize losses immediately and defer the realization of profits (the latter is called a lock-in effect). 
Without further restrictions, there can
even exist a tax-arbitrage  by simultaneously holding both a long and short position in the same risky stock. Accordingly, a loss is realized when it accrues to declare it to the tax office, although the liquidated position is immediately rebuilt (cf., e.g., Constantinides~\cite{constantinides.1983} or Dammon and Green~\cite{dammon.green.1987}).
A popular approach in the literature -- which we also follow in the current paper -- is to exclude this trivial tax arbitrage opportunity
by not allowing for short positions in risky stocks.
At least for retail investors, this is not an unrealistic restriction
(see, e.g., Dybvig and Ross~\cite{dybvig.ross.1986} for a detailed discussion). 

Of course, in practice there are restrictions on negative tax payments and the way in which losses can be offset 
against gains. In some tax systems, losses can only be used to avoid the payment of positive taxes on gains which are realized in the same year.
In other systems, losses can also be carried forwards in time. 
Besides progressive tax rates, these restrictions are another source of nonlinearity that calls for a {\em local} arbitrage theory as developed in Ross~\cite{ross.1987} and Gallmeyer and Srivastava~\cite{gallmeyer.srivastava.2011}. The concepts differ in detail, but the basic idea is
as follows: the investor does not start -- as usual 
in arbitrage theory -- without an endowment and tries to attain a portfolio with nonnegative liquidation value, which
is positive with positive probability. By contrast, an arbitrage is a strategy that is added to another endowment and leads at least to the same liquidation value, albeit with a positive probability to a strictly higher one. 

Another practical restriction is the prohibition of so-called wash sales.
A wash sale is a sale with the aim to declare a loss to the tax office, but the security is immediately repurchased. In the US, a declaration of a loss is not possible if ``substantially identical'' securities are purchased within 30 days after the liquidation of the loss-making security. Under a limited use of losses or the prohibition of wash sales, 
and a positive interest rate, Gallmeyer and Srivastava~\cite{gallmeyer.srivastava.2011} show that in a static Arrow-Debreu security model (i.e., in  a model without redundant securities), no pre-tax arbitrage implies no local after-tax arbitrage. 
Under some parameter restrictions, similar results are obtained for 
the multi-period binomial model including dividends that are taxed at a different rate.

A different approach to the dynamic arbitrage theory with taxes is followed in  Auerbach and Bradford~\cite{auerbach.bradford.2004} and  Jensen~\cite{jensen} who consider more general ``stock value based'' linear tax rules that do not need to be realization based, i.e., taxes may depend on the mark-to-market valuation of stock positions and not only actually-obtained prices. For these tax rules, Jensen~\cite{jensen} characterizes the subset of
``valuation neutral'' tax systems.
This means that he 
starts with a martingale measure for a tax-exempt investor in a multi-period model
and characterizes all ``stock value based'' linear tax systems under which the same measure 
consistently valuates the after-tax wealth of all dynamic trading strategies.
This leads to an interesting decomposition of a linear tax into a tax on the locally riskless interest and a risk sharing component of the government for the ``risky part'' of the investment. 
The investor's gains are decomposed into the interest if her whole capital 
was invested in the riskless bank account and the gains if her actual investments were financed by a short position in the bank account. 
The decomposition depends on the riskless interest rate. This means that in general, the taxes on the gains achieved by investing the entire capital in a risky stock depend on the riskless interest rate. 
These tax systems are also ``holding period neutral'', which means that 
 the holding periods of the shares do not have an impact on investors' selling decisions and consequently, the tax options described above are worthless. 
Especially the permission of short positions in the risky stock does not lead to tax arbitrage. 
A crucial role in \cite{auerbach.bradford.2004} and \cite{jensen}
is played by a tax account for accrued but not yet realized gains, on which interest is paid. The interest payment makes the difference compared with the common real-world tax systems described above. Indeed, the separating measures that are derived in the current 
article depend on the tax rate.
Of course, from a practical perspective, non-realization-based tax systems have many drawbacks. For the mark-to-market valuation, a lot of information is required that is unavailable for less liquid stocks.
In addition, even if taxes on book profits do not need to be payed immediately but 
rather at stock's liquidation time {\em with interest}, it is easy to construct examples in which 
tax liabilities for gains on a stock exceed its liquidation price, which can cause a liquidity problem.\\

In the current article, we discuss the dynamic arbitrage theory under a {\em linear} realization-based capital gains tax. Linearity implies
that we allow for negative tax payments (immediate tax refunds) triggered by the realization of losses. However, it is important to note that our results can still be applied mutatis mutandis to more complicated tax systems. 
On the one hand, no-arbitrage in our model implies no-arbitrage under more restrictive rules. On the other hand, an arbitrage with the full use of losses induces a local arbitrage in the sense of  Definition~1 in Ross~\cite{ross.1987} for, e.g., a model in which gains and losses in a portfolio can only be offset within
the same year. For this, one can consider an investor who receives sufficient annual dividends against which losses can be offset and consequently gains from additional investments are actually taxed linearly.

The restriction to linear taxes implies that in a one-period model, pre-tax and after-tax arbitrage strategies coincide. This allows us to focus the analysis 
on the dynamic component of the problem, i.e., the impact of  tax timing options on the no-arbitrage conditions.  For further discussions on linearity as a ``desiderandum of a tidy tax system'', we refer to  Bradford~\cite{bradford.2000}.

In a dynamic framework, the after-tax no-arbitrage property is  substantially stronger than the pre-tax no-arbitrage property.
Most importantly, in frictionless market models, no-arbitrage is a local property in time. This means that in a discrete time model, the possibility to make a riskless profit can be checked period by 
period, solely based on the current stochastic asset returns,
which is of course an immense reduction of the complexity.
If there does not exist a one-period  
arbitrage -- i.e., an arbitrage strategy that only invests during one predetermined period -- there is also no dynamic arbitrage. 
By contrast, in models with proportional transaction costs or capital gains taxes, no-one-period-arbitrage is strictly weaker than no-arbitrage. In the first case, this is rather obvious since
transaction costs may take some time to amortize (cf., e.g., Example~4.1 in \cite{pham.touzi.1999}). With taxes, assets which allow the deferment of taxes on book profits become more profitable compared with interest- or dividend-paying investments, especially for a long investment horizon
(see, e.g., Black~\cite{black.1976} for a general discussion on the ``sub-optimality of dividend payments'' and \cite{kuehn.ulbricht.2013} for conditions on the stochastic stock price dynamics under which this widely held view holds true). 

The aim of the current article is to {\em quantify} how non-local the no-arbitrage property is in models with capital gains taxes. For this purpose, we introduce the concept of
{\em robust local no-arbitrage} (RLNA) as the weakest local condition which guarantees dynamic no-arbitrage. The condition is local in time, i.e., it can be verified period by period, solely based on the current stochastic stock return and without the knowledge of the stock price returns outside the current period. Robustness refers to  the fact that (RLNA) in the period under consideration guarantees dynamic no-arbitrage whatever the stock price is ``reasonably'' extended outside this period. More precisely, one considers all extensions which are arbitrage-free if the stochastic returns of the period under consideration are eliminated and requires that the stock price process  including this period remains arbitrage-free.
Put differently, this demonstrates which stochastic stock returns in one period can trigger  an arbitrage in a multi-period model. 

On the other hand, no-one-period-arbitrage is also a local property, and it is necessary (rather than sufficient) for dynamic no-arbitrage. This means that no-arbitrage is sandwiched between two local criteria.
Thus, with a sufficient local condition for (RLNA) at hand and   
a characterization of  no-one-period-arbitrage models (which does not depend on the tax rate), we can
estimate how non-local the no-arbitrage property is in models with taxes.   
The sufficient local condition is the main result of the paper (Theorem~\ref{30.10.2017.1}). It can be seen as a generalization of the simple dichotomy condition in a frictionless market with one non-shortable risky stock, which says: given any information at the beginning of 
a period, there is either the risk that the stochastic return falls below the riskless interest rate, or one knows for sure that it does not exceed the riskless interest rate.\\
 
In tax models, the same phenomenon can occur as in models with proportional transaction costs, see Schachermayer~\cite{schachermayer.2004}, which is impossible in frictionless markets by the Dalang-Morton-Willinger theorem~\cite{dmw.1990}:
on an infinite probability space, discrete time no-arbitrage alone does not imply the existence of an equivalent separating measure. 
We provide an explicit example for this in the tax model with a bank account and 
one risky stock (Example~\ref{3.1.2018.1}). In the example, the set of attainable terminal wealths is not closed regarding the convergence in probability.
By contrast, for multi-period proportional 
transaction costs models with only $2$ assets (e.g., a bank account and one risky stock), it is proven by Gigoriev~\cite{grigoriev.2005} that no-arbitrage already implies the existence of a separating measure. 
For a detailed discussion, we refer to the monograph of Kabanov and Safarian~\cite{kabanov.safarian.2009}.  
For our example, we use a price process constructed in Example~\ref{9.11.2017.4}
such that two long positions in the same stock hedge each other. This puzzling phenomenon cannot occur in arbitrage-free frictionless markets or -- more generally -- markets with proportional transactions costs (see Remark~\ref{7.1.2018.1}).
This hedging (im)possibility makes the difference between the situation in
Gigoriev's theorem and Example~\ref{3.1.2018.1}. Using the idea behind Example~\ref{3.1.2018.1}, 
we provide a two-period example for an approximate arbitrage 
in an arbitrage-free transaction costs model with a bank account and 2 risky stocks (see Example~\ref{26.1.2018.2}). We discuss the relation to the original counterexample, Example~3.1 in Schachermayer~\cite{schachermayer.2004}, that demonstrates the same phenomenon in a one-period transaction costs model with $4$ assets.

Finally, we show that the tax model can be written as a model with proportional transaction costs by introducing several 
fictitious securities. We provide a sufficient condition that guarantees the existence of a separating measure, and the set of separating measures is characterized.\\

In the analysis of the paper, we consider the so-called {\em exact tax basis} or {\em specific share identification method}, which corresponds, e.g., to the tax legislation in the US and seems economically
the  most  reasonable  tax  basis.  Here,  an  investor  who  wants  to  reduce  her  position in an asset
can freely choose which of the securities of the same kind in her portfolio (i.e., which Apple stocks)  are relevant for  taxation.  Although  all of these securities possess  the  same  market  price,  they  generally have  different purchasing prices which matters for the taxation.
Other common tax bases are the first-in-first-out rule and an average
of past purchasing prices (see Jouini, Koehl, and Touzi~\cite{jouinikoehltouzi99, jouinikoehltouzi00} and
Ben Tahar, Soner, and Touzi~\cite{bentaharsonertouzi07}, respectively, for solutions of portfolio optimization problems under these tax bases).
Both Theorem~\ref{30.10.2017.1} and Example~\ref{1.12.2017.1}, dealing with the optimality of the dichotomy condition, also hold for the first-in-first-out tax basis and the  average tax basis. The proof of the theorem needs some adjustments, but the method to consider sets like (\ref{1.12.2017.2}) also works here. 
Since this would require a lot of additional notation,
we restrict ourselves to the exact tax basis.\\  

The remainder of this article is organized as follows. In Section~\ref{25.9.2014.1}, we introduce the concept of robust local no-arbitrage~(RLNA), relate it to no-arbitrage~(NA) and
state the main result of the article (Theorem~\ref{30.10.2017.1}). The proofs 
can be found in Section~\ref{29.12.2017.2}. In Section~\ref{29.12.2017.3}, examples of the aforementioned phenomena are provided. In Section~\ref{29.12.2017.4}, the tax model is related to models with proportional transactions costs,
and the set of separating measures is characterized. The article ends with a conclusion.

\section{No-arbitrage and robust local no-arbitrage in the model of Dybvig/Koo}\label{25.9.2014.1}

Throughout the article, we fix a finite time horizon~$T\in\bbn$ and a filtered  probability 
space~$(\Omega,\mathcal{F},(\mathcal{F}_t)_{t=0,1,\ldots,T},P)$.
There is one non-shortable risky stock with price process $S=(S_t)_{t=0,1,\ldots,T}$
where $S_t\in L^0_+(\Omega,\mathcal{F}_t,P)$.
Following the notation in  Dybvig and Koo~\cite{dyb1}, $N_{s,u}\in L^0_+(\Omega,\mathcal{F}_u,P)$ denotes the number of stocks that are bought at time $s\in\{0,\ldots,T\}$ and kept in the portfolio at least after trading 
at time $u\in\{s,\ldots,T\}$. Especially, $N_{s,s}$ is the number of shares purchased at time~$s$, i.e., a position cannot be purchased and resold at the same time. On the 
other hand, a position can be sold and rebought at the same time, which is called a wash sale. For simplicity, we do not exclude wash sales, but this does not have a major impact on the main results as discussed in Remark~\ref{3.1.2018.2}. One has the constraints
\beam\label{18.9.2014.3}
 N_{s,s}\geq N_{s,s+1}\geq\ldots\geq N_{s,T}=0\quad\mbox{for all}~s\in\{0,\ldots,T\},
\eeam
which contains a short-selling restriction and a forced liquidation at $T$.
The random variable $\eta_u\in L^0(\Omega,\mathcal{F}_u,P)$ denotes the number of monetary units {\em after} trading at time~$u$.
There exists a riskless interest rate $r\in\bbr_+\setminus\{0\}$. 
The tax rate satisfies $\alpha\in[0,1)$, i.e., unless otherwise mentioned, the tax-exempt case is included. 
For simplicity, the interest on the bank account is taxed immediately. This means
that the bank account grows with the after-tax interest rate $(1-\alpha)r$.
\begin{Definition}\label{22.11.2017.1}
A process $(\eta,N)$ with $N$ satisfying (\ref{18.9.2014.3}) and $\eta=(\eta_u)_{u=0,1,\ldots,T}$, where $\eta_u\in L^0(\Omega,\mathcal{F}_u,P)$,
is called self-financing for zero initial capital iff 
\beam\label{17.8.2017.2}
\eta_u - \eta_{u-1} =  (1-\alpha)r \eta_{u-1} - N_{u,u}S_u 
+ \sum_{s=0}^{u-1}(N_{s,u-1}-N_{s,u})(S_u-\alpha(S_u-S_s)),\quad u\ge 0,
\eeam
where $\eta_{-1}:=0$.
A self-financing strategy~$(\eta,N)$ with zero initial capital  
is an arbitrage iff $P(\eta_T\ge 0)=1$ and $P(\eta_T> 0)>0$. If no such strategy exists, the market model satisfies no-arbitrage~(NA).
\end{Definition}

\begin{Bemerkung}
If the model satisfies (NA), the same model with a tax-exempt investor (i.e., with $\alpha=0$) also satisfies (NA).
Indeed, if the tax-exempt model allowed for an arbitrage, then there would exist a one-period-arbitrage. This is also a one-period-arbitrage after taxes. But, for $T\ge 2$, the converse does not hold. Consider $S_t=(1+r)^t$. By the deferment of positive taxes, a long position in the stock, along with a short position in the bank account, leads to an arbitrage.
\end{Bemerkung}

\begin{Definition}
Let $t\in\{1,\ldots,T\}$. An arbitrage in the sense of Definition~\ref{22.11.2017.1}
is a one-period-arbitrage in period~$t$ if $N_{s,u}=0$ for all $(s,u)\not=(t-1,t-1)$.
This means that $P(\eta_T\ge 0)=1$ and $P(\eta_T> 0)>0$ with
\beao
\eta_T=(1-\alpha)N_{t-1,t-1}\left(S_t-S_{t-1} -rS_{t-1}\right)(1+(1-\alpha)r)^{T-t}.
\eeao
\end{Definition}
\begin{Notiz}\label{6.2.2018.2}
Let $t\in\{1,\ldots,T\}$ with $P(S_{t-1}>0)=1$.
For $\alpha\in[0,1)$, there is no one-period-arbitrage in period~$t$ iff
\beam\label{30.10.2017.2}
P\left(P\left(\frac{S_t-S_{t-1}}{S_{t-1}} < r\ |\ \mathcal{F}_{t-1}\right)>0\ \mbox{or}\ P\left(\frac{S_t-S_{t-1}}{S_{t-1}} \le r\ |\ \mathcal{F}_{t-1}\right)=1\right) =1.
\eeam
\end{Notiz}
The proof is straightforward. In a one-period model, pre-tax and after-tax arbitrage strategies obviously coincide.
Consequently, (\ref{30.10.2017.2}) is a necessary local condition for dynamic no-arbitrage with taxes, i.e., a condition that only depends on the stochastic return~$(S_t-S_{t-1})/S_{t-1}$ in period~$t$.

\begin{Notiz}\label{6.2.2018.1}
Let $t\in\{1,\ldots,T\}$ and $R\in L^0(\Omega,\mathcal{F}_t,P)$, $R\ge -1$, such that 
$P\left(P\left(R< r\ |\ \mathcal{F}_{t-1}\right)>0\ \mbox{or}\ P\left(R \le r\ |\ \mathcal{F}_{t-1}\right)=1\right) =1$. Then, there exists an adapted process $(S_u)_{u=0,\ldots,T}$ with 
$P(S_{t-1}>0)=1$, $(S_t-S_{t-1})/S_{t-1}=R$, and $S$ satisfies (NA).
\end{Notiz}
\begin{proof}
Consider $S_u=1$ for $u\le t-1$, $S_t=1+R$, and $S_u=0$ for $u\ge t+1$.
\end{proof}
By Note~\ref{6.2.2018.1}, ``no-one-period-arbitrage in $t$'' is the {\em strongest} local property, i.e., a property that only depends on the stochastic return~$(S_t-S_{t-1})/S_{t-1}$ in period $t$, which can be derived from dynamic no-arbitrage with taxes. Put differently, from (NA)
one can only derive properties that are satisfied by all stochastic returns with the ``no-one-period-arbitrage in $t$'' property. Thus, (\ref{30.10.2017.2}) is the best necessary local condition for (NA).

As a counterpart, we introduce a {\em sufficient} local condition that gurantees  (NA). This means for every period~$t$, we again look for a condition that only depends on the stochastic return in period $t$. 

\begin{Definition}[Robust local no-arbitrage (RLNA)]\label{18.8.2017.1}
Given a filtered probability space, a nonnegative adapted stock price process~$S$ satisfies (RLNA) in period~$t\in\{1,\ldots,T\}$ iff 
$P(S_{t-1}=0, S_t>0)=0$ and for all nonnegative adapted processes~$\wh{S}=(\wh{S}_u)_{u=0,1,\ldots,t-1,t+1,\ldots,T}$ acting on the restricted time domain~$\{0,1,\ldots,t-1,t+1,\ldots,T\}$, 
the following implication holds. If $\wh{S}$ satisfies (NA) in the model with restricted time domain~$\{0,1,\ldots,t-1,t+1,\ldots,T\}$ (see Definition~\ref{17.8.2017.3} below for the precise definition of the model),
then $\wt{S}=(\wt{S}_u)_{u=0,1,\ldots,T}$ defined by
\beam\label{17.8.2017.1}
\wt{S}_u 
:= \left\{ \begin{array}{l@{\ :\ }l}    
                         \wh{S}_u   & u\le t-1\\    
                         \wh{S}_{t-1} \frac{S_t}{S_{t-1}} & u=t\\
                         \wh{S}_u\frac{S_t}{S_{t-1}} & u\ge t+1,
                         \end{array}\right.     
\eeam
with the convention $0/0:=0$, also satisfies (NA).

$S$ satisfies (RLNA) iff it satisfies (RLNA) in every period $t\in\{1,\ldots,T\}$.
\end{Definition}
\begin{Bemerkung}
If $P(S_{t-1}=0, S_t>0)>0$, there is a trivial one-period arbitrage with after-tax gain $(1-\alpha)S_t1_{\{S_{t-1}=0\}}$.
Excluding this case, zero is an absorbing state of the asset price and later 
investment opportunities in the asset disappear. 
\end{Bemerkung}

\begin{Bemerkung}\label{7.1.2018.3}
(RLNA) can be interpreted as follows. The stochastic 
return in period~$t$, i.e., $(S_t-S_{t-1})/S_{t-1}$, cannot trigger an arbitrage -- 
regardless of how the process behaves in other periods. To formalize this property, one considers all processes 
with the same stochastic return as $S$ in period $t$ which are arbitrage-free after eliminating the returns of period~$t$. One requires that all of these processes remain arbitrage-free
with period~$t$. The property is local in time since only the return in period~$t$ enters
(although it obviously depends on the time horizon and the whole filtered probability space). 
Put differently, the stock return of period~$t$ pasted together with arbitrary returns forming an arbitrage-free process without period~$t$ should lead to an arbitrage-free process. This justifies the term ``robust''. By construction, (RLNA) is in the sense described above the weakest local condition that guarantees (NA).

For $\alpha=0$, it is equivalent 
to (NA), but for $\alpha>0$, it is nevertheless surprisingly strong (see 
Proposition~\ref{30.10.2017.3}(ii) and Example~\ref{9.11.2017.2}, resp.).\\

To be as general as possible, Definition~\ref{18.8.2017.1} allows that $S$ takes the value zero with positive probability albeit 
that holds minor relevance. For $S>0$, $\wt{S}$ can easily be interpreted as the wealth process of a self-financing portfolio of a tax-exempt investor
buying one stock at price $\wh{S}_0$, selling it at price~$\wh{S}_{t-1}$ to invest the
reward during $t-1$ and $t$ in $S$ and to repurchase $\wh{S}$ at time~$t$. 
\end{Bemerkung}
It remains to formalize the elimination of period~$t$. The idea is that 
the time between $t-1$ and $t$ is eliminated, and consequently
gains accruing in this period -- in both the stock and the bank account -- disappear. On the other hand, the information that is available for the decision on the investment during the next period remains the same as in the original model.

\begin{Definition}[Elimination of period $t$]\label{17.8.2017.3}
The model with eliminated period~$t\in\{1,\ldots,T\}$ and price process 
$(\wh{S}_u)_{u=0,1,\ldots,t-1,t+1,\ldots,T}$ is defined as follows.
A strategy is given by $N=(N_{s,u})_{s,u\in\{0,\ldots,t-1,t+1,\ldots,T\}, s\le u}$
satisfying (\ref{18.9.2014.3}) without time~$t$, 
where $N_{s,u}$ is $\mathcal{F}_u$-measurable if $u\not=t-1$ 
and $N_{s,t-1}$ is $\mathcal{F}_t$-measurable. In the self-financing condition~(\ref{17.8.2017.2}), the increments of $\eta$
between $t-1$ and $t$ and between $t$ and $t+1$ disappear and are replaced by
\beam\label{7.1.2018.2}
\eta_{t+1} - \eta_{t-1} =  (1-\alpha)r \eta_{t-1} - N_{t+1,t+1}\wh{S}_{t+1} 
+ \sum_{s=0}^{t-1}(N_{s,t-1}-N_{s,t+1})(\wh{S}_{t+1}-\alpha(\wh{S}_{t+1}-\wh{S}_s)).
\eeam
(Of course, in the special case $t=T$, one requires $N_{s,T-1}=0$ for $s=0,\ldots,T-1$, and the liquidation value is given by $\eta_{T-1}$ instead of $\eta_T$, i.e., (\ref{7.1.2018.2}) does not apply) 
\end{Definition}
(\ref{7.1.2018.2}) means that the interest between $t-1$ and $t+1$ is  
$r \eta_{t-1}$, i.e., since the time between $t-1$ and $t$ is eliminated, interest is only paid for one time unit. 
Notionally, the investor has to close her stock position at unit price~$\wt{S}_{t-1}$ and repurchase at price $\wt{S}_t$ per share. To make this procedure self-financing, the fraction~$S_{t-1}/S_t$ of the position at time~$t-1$ is repurchased. This motivates the pasting of the price process in (\ref{17.8.2017.1}). 

The assumption that $N_{s,t-1}$ has only to be $\mathcal{F}_t$-measurable is quite natural: the decision on the investment in $\wh{S}_{t+1}-\wh{S}_{t-1}$
can be conditioned on the information~$\mathcal{F}_t$. 
This means that the elimination of the time between $t-1$ and $t$ has no impact on the information 
that is available for the decision on the investment during the next period. For $S>0$, the investment opportunities between $t$ and $t+1$ 
are the same in the markets $\wh{S}$ and $\wt{S}$.

\begin{Bemerkung}
A simpler way to eliminate period $t$ from the model would be that the investor has to liquidate her stock positions at time~$t-1$ and may rebuild them at time~$t$. This would lead to a stronger (RLNA)-condition since in the model without period~$t$, taxes could not be deferred beyond $t-1$. This may even turn sure losses in the stock during period~$t$ into ``good deals''. However, we think that this 
would be the wrong condition since the addition of a period
in the inner of the time domain should not be the reason why taxes can be deferred
over this point. 
\end{Bemerkung}

\begin{Proposition}\label{30.10.2017.3}
We have that 
\begin{itemize}
\item[(i)] For $\alpha\in[0,1)$,\  (RLNA)\ $\Rightarrow$\ (NA)  
\item[(ii)]  For $\alpha=0$,\ (RLNA)\ $\Leftrightarrow$\ (NA)     
\end{itemize}
\end{Proposition}

In the following, we provide a sufficient condition for (RLNA). Example~\ref{9.11.2017.2} shows that it is sharp in some sense.

\begin{thmS}\label{30.10.2017.1}
Let $t\in\{1,\ldots,T\}$ and assume that 
\beam\label{24.6.2018.1}
P(S_{t-1}=0,\ S_t>0)=0
\eeam
and
\beam\label{29.8.2017.1}
P\left(P\left(\frac{S_t-S_{t-1}}{S_{t-1}} < \kappa_{t,T}\ |\ \mathcal{F}_{t-1}\right)>0\ \mbox{or}\ P\left(\frac{S_t-S_{t-1}}{S_{t-1}} \le (1-\alpha)r\ |\ \mathcal{F}_{t-1}\right)=1\right) =1,
\eeam
 where
\beam\label{30.10.2017.4}
\kappa_{t,T} := \frac{\left(-\alpha+(1-\alpha)^2 r\right)(1+(1-\alpha)r)^{T-t} + \alpha}{(1+(1-\alpha)r)^{T-t} - \alpha},
\eeam
again with the convention $0/0:=0$. Then, $S$ satisfies (RLNA) in period~$t$.
\end{thmS}
Condition~(\ref{29.8.2017.1}) is a dichotomy: given 
any information at time~$t-1$, there is either the risk that the stochastic return in period~$t$ falls below $\kappa_{t,T}$ or one knows for sure that it does not exceed  $(1-\alpha)r$.
It turns out that both scenarios make it impossible that the addition of period~$t$  
to an arbitrage-free model induces an arbitrage.

The boundary~$\kappa_{t,T}$ is chosen small enough such that in the case that it is undershot, the loss in period~$t$ dominates the benefit from possible tax defers from $t-1$ to $T$.
Upon first glance, $\kappa_{t,T}$ may look unnecessarily small to guarantee (RLNA);  namely, taxes on gains accrued up to time $t-1$ can also be deferred to $T$ in the comparison model~$\wh{S}$ with eliminated period~$t$, which has to be arbitrage-free by definition. However, since one only requires that $\kappa_{t,T}$ is undershot with positive probability, there can be a gain in period~$t$. This gain could be used to hedge against a bad outcome in the market after $t$ and allow for an arbitrage in the model including period~$t$. 

In Example~\ref{9.11.2017.2}, we construct such a market. 
The example satisfies (\ref{30.10.2017.4}) only with 
a larger boundary~$\kappa<r$ and does not satisfy (RLNA).
\begin{Bemerkung}
Putting $\alpha=0$, (\ref{29.8.2017.1}) reduces to (\ref{30.10.2017.2})
that is necessary and sufficient for ``no-one-period-arbitrage in $t$'', both with and without taxes (see Note~\ref{6.2.2018.2}). 
\end{Bemerkung}
\begin{Bemerkung}
Putting Note~\ref{6.2.2018.1} and Remark~\ref{7.1.2018.3} together, (NA) for some $\alpha\in[0,1)$ is sandwiched between two local conditions, which are the best necessary and sufficient local conditions in the sense described above. In addition, by Theorem~\ref{30.10.2017.1},
(\ref{29.8.2017.1}) guarantees (RLNA), and it is sharp in some sense (see Examples~\ref{9.11.2017.2}).
Thus, comparing the bounds in (\ref{29.8.2017.1}) and (\ref{30.10.2017.2}) provides a  good estimate of how non-local the no-arbitrage property is under taxes (see the conclusion).
\end{Bemerkung}

\begin{Proposition}\label{19.5.2017.1}
Let $\alpha\in(0,1)$ and assume that (\ref{24.6.2018.1}) and
\beam\label{18.5.2017.1}
P\left(\frac{S_t-S_{t-1}}{S_{t-1}}\le (1-\alpha)r\ |\ \mathcal{F}_{t-1}\right)>0,\ \mbox{$P$-a.s.}
\eeam
hold for all $t=1,\ldots,T$, again with the convention $0/0:=0$. Then, the model satisfies (NA).
\end{Proposition}

\begin{Bemerkung}
Condition (\ref{18.5.2017.1}) means that the investor can never be sure that the pre-tax profit in the 
stock strictly exceeds the after-tax profit in the bank account. 
This ensures that the deferment of taxes  on profits in the stock cannot be used to generate an arbitrage. 

On the other hand, for $\alpha=0$ -- which is excluded in 
Proposition~\ref{19.5.2017.1} -- (\ref{18.5.2017.1}) is an even strictly weaker condition than (NA).
\end{Bemerkung}

\section{Proofs}\label{29.12.2017.2}

The following quantities prove useful in all proofs. For every pair $(s,u)$ with $s<u$, 
the after-tax gain at time~$T$ of the self-financing investment in the stock~$S$ between $s$ and $u$ is given by
\beam\label{28.11.2017.2}
X_{s,u} := \left[S_u - \alpha(S_u-S_s)\right](1+(1-\alpha)r)^{T-u} - S_s(1+(1-\alpha)r)^{T-s},\quad s<u.   
\eeam
It is immediate that for any self-financing strategy~$(\eta,N)$ with zero initial capital,
the liquidation value~$\eta_T$, that is uniquely determined by $N$, can be written as 
\beam\label{27.11.2017.1}
V(N):= \eta_T = \sum_{s=0}^{T-1}\sum_{u=s+1}^T (N_{s,u-1}-N_{s,u}) X_{s,u}.
\eeam
For price process~$\wt{S}$ instead of $S$, $\wt{X}_{s,u}$ and $\wt{V}(N)$ are defined accordingly. In the market with smaller time domain~$\{0,\ldots,t-1,t+1,\ldots,T\}$ 
and stock price process~$\wh{S}$ (cf. Definition~\ref{17.8.2017.3}), these quantities read 
\beam\label{22.12.2017.1}
\wh{X}_{s,u} & := & \left[\wh{S}_u - \alpha(\wh{S}_u-\wh{S}_s)\right](1+(1-\alpha)r)^{T-u-1_{(u\le t-1)}}\nonumber\\
& & \qquad  - \wh{S}_s(1+(1-\alpha)r)^{T-s-1_{(s\le t-1)}},\quad s,u\not=t,\ s<u,   
\eeam
and 
\beam\label{27.11.2017.2}
\wh{V}(N) :=
\sum_{s=0,\ s\not=t}^{T-1}\left[\sum_{u=s+1,\ u\not=t,\ u\not=t+1}^T (N_{s,u-1}-N_{s,u})\wh{X}_{s,u} + (N_{s,t-1}-N_{s,t+1})\wh{X}_{s,t+1}\right].
\eeam

\begin{proof}[Proof of Proposition~\ref{30.10.2017.3}]
Ad (i): Assume that $(S_t)_{t=0,\ldots,T}$ does not satisfy (NA) and let 
\beao
t:=\min\left\{u\in\{1,\ldots,T\}\ |\ \exists\ \mbox{arbitrage $(\eta,N)$ with $P(N_{s,l}=0)=1$ for all $l\ge u$}\right\}.
\eeao
This means that $t$ is the smallest number~$u$ such that the model with liquidation time~$u$,
i.e., all stock positions have to be liquidated  up to time~$u$,
allows for an arbitrage.
Let us show that $S$ does not satisfy (RLNA) in period~$t$. W.l.o.g. $P(S_{t-1}=0, S_t>0)=0$, since otherwise (RLNA) does not hold by definition.
We consider the process~$\wh{S}$ acting on the time domain~$\{0,\ldots,t-1,t+1,\ldots,T\}$  and being defined by $\wh{S}_u:=S_u1_{(u\le t-1)}$. 
The corresponding $\wt{S}$ from (\ref{17.8.2017.1})
coincides with $S$ on $\{0,\ldots,t\}$ that allows for an arbitrage. 
Thus, it remains to show that $\wh{S}$ satisfies (NA) in the model from Definition~\ref{17.8.2017.3}. By $S\ge 0$ and $r>0$, one has that $\wh{X}_{s,u}\le (1+(1-\alpha)r)^{-1} X_{s,u\wedge (t-1)}$ for all $s\le t-2$, $u\not= t$, 
$\wh{X}_{t-1,u}\le 0$ for all $u\ge t+1$, and $\wh{X}_{s,u}=0$ for all $s\ge t+1$. For every strategy $N$ in the market~$\wh{S}$, this implies 
\beam\label{24.11.2017.1}
\wh{V}(N) & = & \sum_{s=0}^{t-1} \left[\sum_{u=s+1,\ u\not=t,\ u\not=t+1}^T (N_{s,u-1}-N_{s,u})\wh{X}_{s,u} + (N_{s,t-1}-N_{s,t+1})\wh{X}_{s,t+1}\right]
\nonumber\\
& \le & \frac1{1+(1-\alpha)r} \left[\sum_{s=0}^{t-2}\sum_{u=s+1}^{t-2} (N_{s,u-1}-N_{s,u})X_{s,u}
+ \sum_{s=0}^{t-2} N_{s,t-2} X_{s,t-1}\right],
\eeam
using that $N_{s,T}=0$. The RHS of (\ref{24.11.2017.1}) can be generated in the market with price process~$S$ and liquidation time~$t-1$. But, by the minimality of $t$, $S$ does not allow for an arbitrage if stock positions 
have to be liquidated up to time~$t-1$. Thus, (\ref{24.11.2017.1}) implies (NA)
for $\wh{S}$, and we are done.\\

Ad (ii): It remains to show ``$\Leftarrow$''.  Of course, for $\alpha=0$, (\ref{17.8.2017.2}) reduces to the standard self-financing condition in frictionless markets. Assume that $S$ does not satisfy (RLNA) in some period~$t\in\{1,\ldots,T\}$. If $P(S_{t-1}=0, S_t>0)>0$, $S$ allows for an arbitrage, and we are done. Thus, we can assume that there is a process $\wh{S}$ satisfying (NA), but the corresponding $\wt{S}$ from (\ref{17.8.2017.1}) allows for an arbitrage. Since $\alpha=0$,
there is some $u\in\{1,\ldots,T\}$ such that $\wt{S}$ allows for a one-period-arbitrage 
in period~$u$, i.e., between $u-1$ and $u$ (see the proof of Proposition~5.11 in 
F\"ollmer and Schied~\cite{foellmer.schied.2011}, which also holds under short-selling constraints in some assets). By construction of the model in 
Definition~\ref{17.8.2017.3}, for $u\in\{1,\ldots,T\}\setminus\{t\}$, a 
one-period arbitrage of $\wt{S}$ in $u$ induces a one-period arbitrage of $\wh{S}$
(Note that for $u=t+1$, a one-period arbitrage of $\wt{S}$ would induce an
arbitrage between $t-1$ and $t+1$ in the model with eliminated period~$t$. For this
one needs the relaxation that $N_{s,t-1}$ is only $\mathcal{F}_t$-measurable in the model with $\wh{S}$). But, since $\wh{S}$ satisfies (NA), it follows that $u=t$, 
and $S$ also allows for a one-period arbitrage in $t$.
\end{proof}

\begin{proof}[Proof of Theorem~\ref{30.10.2017.1}]
Let $\wh{S}$ be some arbitrary nonnegative adapted price process satisfying (NA) in the model from Definition~\ref{17.8.2017.3} with the time domain~$\{0,\ldots,t-1,t+1,\ldots,T\}$   and let $\wt{S}$ be the associated process defined in (\ref{17.8.2017.1}) that acts on $\{0,\ldots,T\}$. We observe that
\beam\label{5.7.2018.1}
P\left(\wh{S}_{t-1}=0,\ \wh{S}_u>0\ \mbox{for some\ }u\in\{t+1,\ldots,T\}\right)=0.
\eeam
Furthermore, throughout the proof, we fix a strategy~$N$ in the market~$\wt{S}$ 
on the time domain~$\{0,\ldots,T\}$. We have to show that $N$ cannot be an arbitrage,
i.e., $\wt{S}$ also satisfies (NA).\\ 

{\em Step 1:} Define 
\beao
B_1:= \left\{ \wh{S}_{t-1}=0 \right\}\cup\left\{P\left(\frac{S_t-S_{t-1}}{S_{t-1}} \le (1-\alpha)r\ |\ \mathcal{F}_{t-1}\right)=1\right\}.
\eeao 
Let us first show that there exists a strategy $\wh{N}$ in the market~$\wh{S}$ on
$\{0,\ldots,t-1,t+1,\ldots,T\}$ (cf. Definition~\ref{17.8.2017.3}) such that $\wh{N}_{s,u} = N_{s,u}$ for all $s\le u\le t-2$ and
\beam\label{12.9.2017.1}
\wh{V}(\wh{N}) \ge \frac1{1+(1-\alpha)r}\wt{V}(N)\quad\mbox{on}\quad B_1\cup
\left\{\sum_{s=0}^{t-1}N_{s,t-1}=0\right\}\quad P\mbox{-a.s.},
\eeam
where $\wh{V}$ and $\wt{V}$ are defined as in (\ref{27.11.2017.2}) and  (\ref{27.11.2017.1}), respectively. We define $(\wh{N}_{s,u})_{s\le u,\ s,u\not=t}$ by
\beao 
\wh{N}_{s,u} & := & N_{s,u},\quad s\le t-2,\ u\not=t-1\\
\wh{N}_{s,t-1} & := & N_{s,t},\quad s\le t-2\\
\wh{N}_{t-1,t-1} & := & N_{t-1,t} +  \frac1{1+(1-\alpha)r}\frac{S_t}{S_{t-1}}N_{t,t}\\
\wh{N}_{t-1,u} & := & N_{t-1,u} +  \frac1{1+(1-\alpha)r}\frac{S_t}{S_{t-1}}N_{t,u},\quad u\ge t+1\\
\wh{N}_{s,u} & := & \frac1{1+(1-\alpha)r}\frac{S_t}{S_{t-1}}N_{s,u},\quad s\ge t+1.
\eeao
Note that $\wh{N}_{s,t-1}$ has only to be $\mathcal{F}_t$-measurable.
$\wh{N}_{t-1,t-1}$ is the number of stocks which are purchased at price~$\wh{S}_{t-1}$, i.e.,``between $t-1$ and $t+1$'' in the model with the smaller time domain. These purchases have to mimic the sum of purchases at price  $\wt{S}_{t-1}$ and $\wt{S}_t$ in the model with the larger time domain.

On the set $\left\{\sum_{s=0}^{t-1}N_{s,t-1}=0\right\}$, the stock positions in the market~$\wt{S}$ are completely liquidated at $t-1$
and no new shares are purchased at time~$t-1$. (\ref{18.9.2014.3}) yields 
that on this set one has $N_{s,u}=0$ for all $s\le t-1$, $u\ge t-1$, and thus, $\wt{V}(N)$ reduces to
\beam\label{23.6.2018.1}
\wt{V}(N) = \sum_{s=0}^{t-2}\sum_{u=s+1}^{t-1} (N_{s,u-1}-N_{s,u}) \wt{X}_{s,u} + \sum_{s=t}^{T-1}\sum_{u=s+1}^T (N_{s,u-1}-N_{s,u}) \wt{X}_{s,u}.
\eeam
By construction of $\wh{N}$, on the above set, one has
$\wh{N}_{s,u}=0$ for all $s\le t-2$, $u\ge t-1$, 
$\wh{N}_{t-1,t-1}=1/(1+(1-\alpha)r)S_t/S_{t-1} N_{t,t}$, and
$\wh{N}_{t-1,u}=1/(1+(1-\alpha)r)S_t/S_{t-1} N_{t,u}$ for all $u\ge t+1$.
This yields a similar simplification as in (\ref{23.6.2018.1}):
\beam\label{23.6.2018.2}
\wh{V}(\wh{N}) & = & \sum_{s=0}^{t-2}\sum_{u=s+1}^{t-1} (\wh{N}_{s,u-1}-\wh{N}_{s,u}) \wh{X}_{s,u} 
+ \sum_{u=t+1}^T \frac1{1+(1-\alpha)r}\frac{S_t}{S_{t-1}}(N_{t,u-1}-N_{t,u}) \wh{X}_{t-1,u}\nonumber\\
& & + \sum_{s=t+1}^{T-1} \sum_{u=s+1}^T (\wh{N}_{s,u-1}-\wh{N}_{s,u}) \wh{X}_{s,u}.
\eeam 
For $s\le t-2$, $u\le t-1$, one has $\wh{N}_{s,u-1}-\wh{N}_{s,u} = N_{s,u-1}-N_{s,u}$, $\wh{X}_{s,u}=\wt{X}_{s,u}/(1+(1-\alpha)r)$,  
and thus
$(\wh{N}_{s,u-1}-\wh{N}_{s,u})\wh{X}_{s,u}
=(N_{s,u-1}-N_{s,u})\wt{X}_{s,u}/(1+(1-\alpha)r)$. 
On the other hand, 
$\wh{N}$ yields the same gain in the submarket~$\wh{S}$ on $\{t-1,t+1,t+2\ldots,T\}$ 
as $N$ in the submarket~$\wt{S}$ on $\{t,t+1,t+2,\ldots,T\}$ up to the prefactor~$1/(1+(1-\alpha)r)$.  
Since both gains disappear on the set~$\{S_t=0\}$, we only have to check this assertion on the set~$\{S_t>0\}$, which coincides $P$-a.s. with $\{S_{t-1}>0,S_t>0\}$ by assumption. For $s\ge t+1$, $u>s$, we have $\wh{N}_{s,u-1}-\wh{N}_{s,u}  =
1/(1+(1-\alpha)r)S_t/S_{t-1}(N_{s,u-1}-N_{s,u})$, 
$\wh{X}_{s,u}=S_{t-1}/S_t \wt{X}_{s,u}$, and thus
$(\wh{N}_{s,u-1}-\wh{N}_{s,u})\wh{X}_{s,u}  =
1/(1+(1-\alpha)r) (N_{s,u-1}-N_{s,u})\wt{X}_{s,u}$. 
In addition, one has $\wh{X}_{t-1,u}=S_{t-1}/S_t\wt{X}_{t,u}$.
Putting together, the ``corresponding'' summands in (\ref{23.6.2018.1}) and (\ref{23.6.2018.2}) coincide up to the prefactor~$1/(1+(1-\alpha)r)$, which implies that (\ref{12.9.2017.1}) is satisfied with equality on $\left\{\sum_{s=0}^{t-1}N_{s,t-1}=0\right\}$.

Now, we analyze the gains on $B_1$, which is of course the interesting part. Without the assumption that $\sum_{s=0}^{t-1}N_{s,t-1}=0$, one
needs estimates for the gains~$\wt{X}_{s,u}$ with $s\le t-1$ and $u\ge t$.
We obtain that
\beam\label{30.11.2017.1}
\wt{X}_{s,u} & = & \left[(1-\alpha)\wt{S}_u + \alpha \wt{S}_s\right](1+(1-\alpha)r)^{T-u} - \wt{S}_s(1+(1-\alpha)r)^{T-s}\nonumber\\
& \le & \left[(1-\alpha)(1+(1-\alpha)r)\wh{S}_u + \alpha \wh{S}_s\right](1+(1-\alpha)r)^{T-u} - \wh{S}_s(1+(1-\alpha)r)^{T-s}\nonumber\\
& \le & (1+(1-\alpha)r)\left(\left[(1-\alpha)\wh{S}_u + \alpha \wh{S}_s\right](1+(1-\alpha)r)^{T-u} - \wh{S}_s(1+(1-\alpha)r)^{T-s-1}\right)\nonumber\\
& = & (1+(1-\alpha)r)\wh{X}_{s,u}\quad\mbox{$P$-a.s. on}\ B_1,\quad s\le t-1,\ u\ge t+1,
\eeam
in which for the estimate on the set $\{\wh{S}_{t-1}=0\}$ we use that
$\{\wh{S}_{t-1}=0\}\subset \{\wh{S}_u=0,\ \wt{S}_u=0\}$,\ $P$-a.s. by (\ref{5.7.2018.1}) and the construction of $\wt{S}$. With the same estimation, it follows
\beam\label{30.11.2017.2}
\wt{X}_{s,t} 
& \le & \left[(1-\alpha)(1+(1-\alpha)r)\wh{S}_{t-1} + \alpha \wh{S}_s\right](1+(1-\alpha)r)^{T-t} - \wh{S}_s(1+(1-\alpha)r)^{T-s}\nonumber\\
& \le & (1+(1-\alpha)r)\wh{X}_{s,t-1}\quad\mbox{$P$-a.s. on}\ B_1,\quad s\le t-2
\eeam
and
\beam\label{30.11.2017.3}
\wt{X}_{t-1,t}\le 0\quad\mbox{$P$-a.s. on}\ B_1.
\eeam 
Without the assumption that $\sum_{s=0}^{t-1}N_{s,t-1}=0$, the RHSs of 
(\ref{23.6.2018.1}) and (\ref{23.6.2018.2}) still coincide up to the prefactor~$1/(1+(1-\alpha)r)$, but there appear the additional summands 
\beam\label{6.7.2018.1}
\sum_{s=0}^{t-1}\sum_{u=t}^T (N_{s,u-1}-N_{s,u}) \wt{X}_{s,u}
\eeam
and
\beam\label{6.7.2018.2}
\sum_{s=0}^{t-1}\sum_{u=t+1}^T (N_{s,u-1}-N_{s,u}) \wh{X}_{s,u}
+  \sum_{s=0}^{t-2}(N_{s,t-1}-N_{s,t}) \wh{X}_{s,t-1}
\eeam
for $\wt{V}(N)$ and $\wh{V}(\wh{N})$, respectively. By (\ref{30.11.2017.1}) and (\ref{30.11.2017.2}), 
each summand in (\ref{6.7.2018.2}) dominates its
``corresponding'' summand in (\ref{6.7.2018.1}) up to the 
prefactor~$1/(1+(1-\alpha)r)$ on the set $B_1$. The summand 
$(N_{t-1,t-1}-N_{t-1,t}) \wt{X}_{t-1,t}$ in (\ref{6.7.2018.1}) is left,
but by (\ref{30.11.2017.3}) it is nonpositive on $B_1$, and we arrive at 
(\ref{12.9.2017.1}).

This means that if the {\em pre-tax} stock return in period $t$ certainly does not exceed $(1-\alpha)r$,  
an elimination of period~$t$ would always be desirable for the investor.
Note that an elimination means that she can defer taxes without being invested in period~$t$, and she need not pay interest on her debts in period~$t$.\\

{\em Step 2:} Define
\beao
B_2:= \left\{\wh{S}_{t-1}>0\right\}\cap\left\{S_{t-1}>0,\ P\left(\frac{S_t-S_{t-1}}{S_{t-1}} < \kappa_{t,T}\ |\ \mathcal{F}_{t-1}\right)>0\right\}.
\eeao 
Note that by (\ref{24.6.2018.1}) and the convention $0/0:=0$, we have that $\{S_{t-1}=0\}\subset B_1$,\ $P$-a.s. Thus, we get by assumption~(\ref{29.8.2017.1}) that 
\beam\label{5.7.2018.2}
P(B_1\cup B_2)=1.
\eeam

On $B_2$ the following decomposition plays a crucial role. For $s\le t-1$, $u\ge t$, we decompose $\wt{X}_{s,u}$
into four parts:
the gain when liquidating the stock at $t-1$, the gain after repurchasing the stock 
at time $t$, the
wealth generated by deferring the tax on gains accrued up to time~$t-1$ to 
time~$u$, and the profit in period $t$ taxed at time $u$. Formally, the decomposition is also defined for $s<u$ with $s>t-1$ or $u<t$, but then it degenerates, i.e., $I_3^{s,u}=I_4^{s,u}=0$.
The decomposition reads $\wt{X}_{s,u} = I_1^{s,u} + I_2^{s,u} + I_3^{s,u} + I_4^{s,u}$, where
\beao
I_1^{s,u} := \left[(1-\alpha)\wt{S}_{u\wedge (t-1)} + \alpha \wt{S}_{s\wedge(t-1)}\right]  
\left( 1+(1-\alpha)r\right)^{T-u\wedge(t-1)}
- \wt{S}_{s\wedge(t-1)}\left( 1+(1-\alpha)r\right)^{T-s\wedge(t-1)},
\eeao

\beao
I_2^{s,u} := \left[(1-\alpha)\wt{S}_{u\vee t} + \alpha \wt{S}_{s\vee t}\right]  
\left( 1+(1-\alpha)r\right)^{T-u\vee t}
- \wt{S}_{s\vee t}\left( 1+(1-\alpha)r\right)^{T-s\vee t},
\eeao

\beao
I_3^{s,u} := \alpha\left[\wt{S}_{u\wedge (t-1)} - \wt{S}_{s\wedge(t-1)}\right]  
\left[\left( 1+(1-\alpha)r\right)^{T-u\wedge(t-1)}
- \left( 1+(1-\alpha)r\right)^{T-u}\right],
\eeao

$I_4^{s,u}:=0$ for $s> t-1$ or $u< t$ and otherwise

\beao
I_4^{s,u} :=
\wt{S}_t\left( 1+(1-\alpha)r\right)^{T-t}
- \wt{S}_{t-1}\left( 1+(1-\alpha)r\right)^{T-(t-1)}
-\alpha(\wt{S}_t-\wt{S}_{t-1})\left( 1+(1-\alpha)r\right)^{T-u}.
\eeao

By $\wt{S}\ge 0$, we have that
\beam\label{14.1.2018.1}
& & I_3^{s,u} + I_4^{s,u}\\
& & \le \alpha \wt{S}_{t-1}   
\left[\left(1+(1-\alpha)r\right)^{T-(t-1)}
- (1+(1-\alpha)r)^{T-u}\right]\nonumber\\
& & + \wt{S}_t\left( 1+(1-\alpha)r\right)^{T-t}
- \wt{S}_{t-1}\left( 1+(1-\alpha)r\right)^{T-(t-1)}
-\alpha(\wt{S}_t-\wt{S}_{t-1})\left( 1+(1-\alpha)r\right)^{T-u}\nonumber
\eeam
for all $s\le t-1$, $u\ge t$.
By $\wt{S}_t\ge 0$, the RHS of (\ref{14.1.2018.1}) takes its maximum at $u=T$, which implies
\beam\label{30.8.2017.1}
& & I_3^{s,u} + I_4^{s,u}\\
& & \le I\nonumber\\
& & := \alpha \wt{S}_{t-1}   
\left[\left(1+(1-\alpha)r\right)^{T-(t-1)}
- 1\right]\nonumber\\
& & + \wt{S}_t\left( 1+(1-\alpha)r\right)^{T-t}
- \wt{S}_{t-1}\left( 1+(1-\alpha)r\right)^{T-(t-1)}
-\alpha(\wt{S}_t-\wt{S}_{t-1})\nonumber\\
& & \le \alpha \wt{S}_{t-1}  
\left[\left(1+(1-\alpha)r\right)^{T-(t-1)}
- 1\right] \nonumber\\
& &  + \wt{S}_{t-1}\left[(1+\kappa_{t,T})\left( 1+(1-\alpha)r\right)^{T-t}
- \left( 1+(1-\alpha)r\right)^{T-(t-1)}
-\alpha\kappa_{t,T}\right]\nonumber\\
& & < 0 \qquad\qquad\qquad\qquad
\quad\mbox{on}\ \left\{\wh{S}_{t-1}>0,\ S_{t-1}>0,\ \frac{S_t-S_{t-1}}{S_{t-1}} < \kappa_{t,T}\right\}\ 
\mbox{for}\ s\le t-1,\ u\ge t\nonumber.
\eeam
(\ref{30.8.2017.1}) can be seen as the key estimate of the proof. It implies that
on the event~$B_2\in\mathcal{F}_{t-1}$, there is the risk that the loss in period~$t$ dominates the benefit from deferring taxes over period~$t$. 
The estimate holds simultaneously in $s\in\{0,\ldots,t-1\}$ and $u\in\{t,\ldots,T\}$. 

Now define 
\beao
V_i := \sum_{(s,u),\ s<u} (N_{s,u-1}-N_{s,u})I_i^{s,u},\quad i=1,2,3,4.
\eeao
The terminal wealth $\wt{V}(N)$ is given by $\wt{V}(N) = V_1 + V_2 + V_3 + V_4$. 
First note that $V_1$ is $\mathcal{F}_{t-1}$-measurable, which can be seen by 
writing it as
\beao
V_1= \sum_{s,\ s<t-1}\left(\sum_{u,\ u>s,\ u<t-1}(N_{s,u-1}-N_{s,u})I_1^{s,u} 
+ N_{s,t-2}I_1^{s,t-1}\right).
\eeao
We also consider
\beao
W:= \sum_{(s,u),\ s\le t-1,\ u\ge t} (N_{s,u-1}-N_{s,u})I = I \sum_{s,\ s\le t-1} N_{s,t-1}. 
\eeao
By (\ref{30.8.2017.1}), one has $W\ge V_3+V_4$ everywhere, and  in contrast to $V_3+V_4$, $W$ is $\mathcal{F}_t$-measurable.\\ 

{Step 3:} Now, we prepare a case differentiation to complete the proof. Define
\beam\label{1.12.2017.2}
& & \wh{\mathcal{M}}^N := \left\{ A\in\mathcal{F}_{t-1}\ |\ \exists \wh{N}\ \mbox{on}\ \{0,\ldots,t-1,t+1,\ldots,T\}\ \mbox{such that}\right.\\  
& & \qquad\qquad\left. \wh{N}_{s,u} = N_{s,u}\ \mbox{$P$-a.s.}\ \forall  s\le u\le t-2\ \mbox{and}\ P(\wh{V}(\wh{N})\ge 0\ |\ \mathcal{F}_{t-1})=1\ \mbox{on}\ A\ P\mbox{-a.s.}\right\},\nonumber 
\eeam
\beao
\mathcal{M}^N := \left\{ A\in\mathcal{F}_{t-1}\ |\ \exists \wt{N}\ \mbox{on}\ \{0,\ldots,T\}\ \mbox{such that}\ 
\wt{N}_{s,u} = N_{s,u}\ \mbox{$P$-a.s.}\ \forall s\le u\le t-2\right.\\ 
\left. \mbox{and}\ P(\wt{V}(\wt{N})\ge 0\ |\ \mathcal{F}_{t-1})=1\ \mbox{on}\ A\ P\mbox{-a.s.}\right\}, 
\eeao
\beam\label{19.9.2017.1}
\wh{M}^N:={\rm ess sup}\ \wh{\mathcal{M}}^N\quad (\mbox{i.e.,}\ 1_{\wh{M}^N} 
= {\rm ess sup}\  \{1_A\ |\ A\in\wh{\mathcal{M}}^N\}),\quad
\mbox{and}\quad M^N:={\rm ess sup}\ \mathcal{M}^N
\eeam
Of course, the essential supremum of the family of 
functions~$\{1_A\ |\ A\in\wh{\mathcal{M}}^N\}$ is $\{0,1\}$-valued, which allows
the definitions~(\ref{19.9.2017.1}), cf., e.g., Remark~1.14 of \cite{he.wang.yan.1992}.

Let us show that the suprema in (\ref{19.9.2017.1}) are attained, i.e., 
\beam\label{1.12.2017.1}
\wh{M}^N\in \wh{\mathcal{M}}^N
\eeam
(and of course the same with $M^N$, although not needed). Indeed, by general properties of the essential supremum, there exists a sequence $(A_n)_{n\in\bbn}\subset \wh{\mathcal{M}}^N$ such that $\cup_{n\in\bbn} A_n=\wh{M}^N$\ $P$-a.s (cf., e.g., Remark~1.14 of \cite{he.wang.yan.1992}). Let $\wt{N}^{(n)}$ be corresponding strategies with
$\wt{N}^{(n)}_{s,u} = N_{s,u}$\ $P$-a.s. for all $s\le u\le t-2$  
and $P(\wt{V}(\wt{N}^{(n)})\ge 0\ |\ \mathcal{F}_{t-1})=1$ on $A_n$\ $P$-a.s. Now, 
one paste these strategies together by defining
$\wt{N}_{s,u}:=N_{s,u}$ for $u\le t-2$, 
$\wt{N}_{s,u} := \sum_{n=1}^\infty 1_{A_n\setminus(A_1\cup\ldots\cup A_{n-1})}N^{(n)}_{s,u}$ for $u\ge t-1$ (and of course, $s\le u$). This yields (\ref{1.12.2017.1}).\\

The set~$\wh{M}^N\in\mathcal{F}_{t-1}$ is the event that the strategy~$N$ ``before $t-1$'' can be extended to a strategy in the market~$\{0,\ldots,t-1,t+1,\ldots,T\}$ that does not make a loss. In arbitrage-free frictionless markets, this condition  would be equivalent to the non-negativity of the liquidation value at $t-1$. But, by the deferment of taxes, it can happen that a negative liquidation value becomes positive for sure with the passing of time. Note that in (\ref{1.12.2017.2}), one has $u\le t-2$
and not $u\le t-1$. This means that given the information~$\mathcal{F}_{t-1}$, trades
at time~$t-1$ can differ from $N$, e.g., all stock positions can be liquidated at $t-1$.\\

From (\ref{12.9.2017.1}) and $B_1\in\mathcal{F}_{t-1}$, it follows that
\beam\label{12.9.2017.2}
M^N \cap B_1 \subset \wh{M}^N \cap B_1\quad P\mbox{-a.s.}
\eeam

Now, we distinguish four cases, that may overlap, but include everything due to (\ref{5.7.2018.2}), to show that $N$ cannot be an 
arbitrage.\\

{\em Case 1:} $P\left(B_2\cap \left\{\sum_{s=0}^{t-1}N_{s,t-1}>0\right\}\right) =0$. 

By (\ref{5.7.2018.2}), one has $P\left(B_1\cup\left\{\sum_{s=0}^{t-1}N_{s,t-1}=0\right\}\right) =1$. Then, by (\ref{12.9.2017.1}),
there exists an $\wh{N}$ in the market~$\{0,\ldots,t-1,t+1,\ldots,T\}$
with $P(\wh{V}(\wh{N})\ge \wt{V}(N)/(1+(1-\alpha)r))=1$. Since the market~$\{0,\ldots,t-1,t+1,\ldots,T\}$
satisfies (NA), $N$ cannot be an arbitrage.\\

{\em Case 2:} $P\left(B_2\cap \left\{\sum_{s=0}^{t-1} N_{s,t-1}>0\right\}\right)>0$ and $P(\wh{M}^N)=1$.

By (\ref{1.12.2017.1}), there exists an $\wh{N}$ with 
$\wh{N}_{s,u} = N_{s,u}\ \mbox{$P$-a.s.}$ for all  $s\le u\le t-2$ and $P(\wh{V}(\wh{N})\ge 0)=1$. If the event~$\{V_1>0\}\in\mathcal{F}_{t-1}$ had a positive probability, the 
strategy~$\wh{N}'$ defined by $\wh{N}'_{s,u} := N_{s,u}$ for $u\le t-2$ and
$\wh{N}'_{s,u} := 1_{\{ V_1\le 0\}} \wh{N}_{s,u}$ for $u\ge t-1$ would be an arbitrage since $\wh{V}(\wh{N}') = 1_{\{ V_1 >0\}} V_1 + 1_{\{ V_1\le 0\}} \wh{V}(\wh{N})$. Thus, since the market~$\{0,\ldots,t-1,t+1,\ldots,T\}$
is arbitrage-free by assumption, we must have that $P(V_1\le 0)=1$. By $\{\sum_{s=0}^{t-1} N_{s,t-1}>0\}\in\mathcal{F}_{t-1}$ and (\ref{30.8.2017.1}), one has 
\beam\label{29.8.2017.2}
0 < P\left(\sum_{s=0}^{t-1} N_{s,t-1}>0,\ \wh{S}_{t-1}>0,\ S_{t-1}>0,\ \frac{S_t-S_{t-1}}{S_{t-1}} < \kappa_{t,T}\right) \le P(W<0).
\eeam
Since $\{W<0\}\in\mathcal{F}_t$, the random gain $1_{\{W<0\}}V_2$ can also be generated in the submarket with price process~$(\wh{S}_u)_{u=t-1,t+1,t+2,\ldots,T}$, 
in which initial purchases have only to be $\mathcal{F}_t$-measurable.
Since this submarket is arbitrage-free, (\ref{29.8.2017.2}) implies that 
$P(W+V_2<0)>0$ and thus $P(\wt{V}(N)<0)>0$.\\

{\em Case 3:} $P((\Omega\setminus \wh{M}^N)\cap B_2)>0$.

By choosing  $\wh{N}_{s,u} := N_{s,u}$ for $u\le t-2$ and
$\wh{N}_{s,u}:=0$ for $u\ge t-1$, it can be seen that $\{V_1\ge 0\}\in \wh{\mathcal{M}}^N$. Thus, $\Omega\setminus \wh{M}^N\subset \{V_1<0\}$ $P$-a.s.
and $P(\{V_1<0\}\cap B_2)>0$.
This together with $\{V_1<0\}\in\mathcal{F}_{t-1}$ and (\ref{30.8.2017.1})  
implies that
\beao
0 < P\left(V_1<0,\ \wh{S}_{t-1}>0,\ S_{t-1}>0,\  \frac{S_t-S_{t-1}}{S_{t-1}} < \kappa_{t,T}\right) 
\le P(V_1<0, W\le 0).
\eeao
Then, again by the $\mathcal{F}_t$-measurability of $V_1+W$ and the no-arbitrage
property of the  submarket with price process~$(\wh{S}_u)_{u=t-1,t+1,t+2,\ldots,T}$, 
one arrives at $P(V_1+W+V_2<0)>0$ and thus $P(\wt{V}(N)<0)>0$.

{\em Case 4:} $P((\Omega\setminus \wh{M}^N)\cap B_1)>0$.

By (\ref{12.9.2017.2}), one has $P(\Omega\setminus M^N)>0$.
On the other hand, by the maximality of $M^N$, one has $\Omega\setminus M^N \subset 
\{ P(\wt{V}(N)<0\ |\ \mathcal{F}_{t-1}) >0 \}$\ $P$-a.s.
Putting together, we arrive at $P(\wt{V}(N)<0)>0$.
\end{proof}

\begin{proof}[Proof of Proposition \ref{19.5.2017.1}]
Assume that (\ref{24.6.2018.1}) and (\ref{18.5.2017.1}) hold.\\ 

{\em Step 1:} Define $A_s:= \{S_t\le (1+(1-\alpha)r)S_{t-1},\ \forall 
t=s+1,\ldots,T\}$. Let us show that 
\beam\label{18.5.2017.2}
P(A_s\ |\ \mathcal{F}_s)>0\quad P\mbox{-a.s.}
\eeam
by backward-induction in $s=T-1,T-2,\ldots,0$. For $s=T-1$, the assertion is already included in
(\ref{18.5.2017.1}). $s\leadsto s-1$: We have 
$A_{s-1}=A_s\cap\{S_s\le (1+(1-\alpha)r)S_{s-1}\}$. Let $C\in\mathcal{F}_{s-1}$ with $P(C)>0$. By (\ref{18.5.2017.1}), this implies that 
\beam\label{28.11.2017.1}
& & P(C\cap \{S_s\le (1+(1-\alpha)r)S_{s-1}\})\nonumber\\
& & = E\left(E\left(1_C 1_{\{S_s\le (1+(1-\alpha)r)S_{s-1}\}}\ |\ \mathcal{F}_{s-1}\right)\right)\nonumber\\
& & = E\left(1_C P(S_s\le (1+(1-\alpha)r)S_{s-1}\ |\ \mathcal{F}_{s-1})\right)
>0.
\eeam
Together with $C\cap \{S_s\le (1+(1-\alpha)r)S_{s-1}\}\in\mathcal{F}_s$ and
the induction hypothesis, (\ref{28.11.2017.1}) implies that 
\beao
P(C\cap A_{s-1})=E(1_{C\cap\{S_s\le (1+(1-\alpha)r)S_{s-1}\}}P(A_s\ |\ \mathcal{F}_s))>0,
\eeao
and we are done.\\

Let $s<t$. On $A_s$, the liquidation value at time~$t$ of a stock purchased at time~$s$ satisfies  
\beao
S_t - \alpha(S_t-S_s) & \le & S_s\left( 1+(1-\alpha)r\right)^{t-s} 
- \alpha \left( (1+(1-\alpha)r)^{t-s}-1\right)S_s\nonumber\\ 
& \le & S_s\left( 1+(1-\alpha)r\right)^{t-s} 
- \alpha(1-\alpha)rS_s,
\eeao
and thus
\beam\label{18.5.2017.3}
X_{s,t}\le -\alpha(1-\alpha)rS_s(1+(1-\alpha)r)^{T-t} < 0\quad\mbox{on}\ A_s\cap\{S_s>0\}\ \mbox{for all\ }t\ge s+1,
\eeam
where $X$ is defined in (\ref{28.11.2017.2}). On the other hand, by
(\ref{24.6.2018.1}), one has
\beam\label{24.6.2018.2}
X_{s,t}=0\quad\mbox{on $\{S_s=0\}$\ $P$-a.s. for all $t\ge s+1$}.
\eeam

{\em Step 2:} Now, let $N$ be some arbitrary strategy in the stock with liquidation value $V(N)$ from (\ref{27.11.2017.1}). Define the stopping time 
\beao
\tau:=\inf\{s\ge 0\ |\ N_{s,s}>0\ \mbox{and}\ S_s>0\}\wedge T.
\eeao 

{\em Case 1:} $P(\tau=T)=1$. Either the strategy does not trade at all or only at a vanishing stock price. To see this, define 
$\tau':=\inf\{s\ge 0\ |\ N_{s,s}>0\}\wedge T$. We have $S_{\tau'}=0$, $P$-a.s. on $\{\tau'<T\}$. By (\ref{24.6.2018.1}), this implies that $S_t=0$ for all $t=\tau',\tau'+1,\ldots,T$, $P$-a.s. on $\{\tau'<T\}$. Thus, $N$
satisfies $(N_{s,u-1}-N_{s,u})X_{s,u}=0$ for all $s=0,\ldots,T-1,\ u=s+1,\ldots,T$, $P$-a.s. and cannot be an arbitrage.\\

{\em Case 2:} $P(\tau=T)<1$. (\ref{18.5.2017.2}) implies that
\beao
P(A)>0,\quad \mbox{where\ }A:=\{\tau< T\}\cap\{S_t\le (1+(1-\alpha)r)S_{t-1}\quad\forall t=\tau+1\ldots,T\}.
\eeao
Note that $\{\tau< T\}\subset\{N_{\tau,\tau}>0,\ S_\tau>0\}$. By (\ref{18.5.2017.3}), we get $X_{\tau(\omega),t}(\omega)<0$ for all $t\ge \tau(\omega)+1$ and $\omega\in A$. Together with $N_{\tau,t-1}-N_{\tau,t} \ge 0$ for all $t\ge \tau + 1$ and
$\sum_{t=\tau+1}^T (N_{\tau,t-1}-N_{\tau,t})=N_{\tau,\tau}$, this implies that
\beam\label{24.6.2018.3}
\sum_{t=\tau(\omega)+1}^T (N_{\tau(\omega),t-1}(\omega)-N_{\tau(\omega),t}(\omega))X_{\tau(\omega),t}(\omega)<0\quad\mbox{for all\ }\omega\in A.
\eeam
On the other hand, we have that for all pairs~$(s,t)$ with $s<t$
\beam\label{6.11.2017.1}
(N_{s,t-1}-N_{s,t})X_{s,t}\le 0\quad\mbox{on\ }A,\quad P\mbox{-a.s.}
\eeam
Indeed, by (\ref{24.6.2018.2}), it remains to consider the case that 
$S_s(\omega)>0$. If in addition $N_{s,t-1}(\omega)-N_{s,t}(\omega)>0$, then 
$\tau(\omega)\le s$ and (\ref{6.11.2017.1}) follows from (\ref{18.5.2017.3}).
From (\ref{24.6.2018.3}) and (\ref{6.11.2017.1}), one obtains $V(N)<0$ on $A$. Thus, $N$ cannot be an arbitrage.
\end{proof} 

\section{(Counter-)Examples}\label{29.12.2017.3}

In the examples, we have $\alpha\in(0,1)$, $\mathcal{F}=2^\Omega$, and all states have a positive probability. In addition, the following simple observations prove useful in many places. 
\begin{Notiz}\label{26.6.2018.5}

\begin{itemize}
\item[(i)] Let $R\in\bbr_+$ and the real number $\bar{r}$ is given by
\beam\label{27.6.2018.1}
(1+R)(1+\bar{r})(1-\alpha) + \alpha = \left[(1+R)(1-\alpha)+\alpha\right](1+(1-\alpha)r).
\eeam
Then, one has $\bar{r}\in ((1-\alpha)r,r]$, where $\bar{r}=r$ iff $R=0$, and for every $R'>R$, 
\beam\label{27.6.2018.2}
(1+R')(1+\bar{r})(1-\alpha) + \alpha > \left[(1+R')(1-\alpha)+\alpha\right](1+(1-\alpha)r).
\eeam
\item[(ii)] Let $n\in\bbn_0$ and the real number $R$ is given by
\beam\label{26.6.2018.1}
(1+R)(1-\alpha) + \alpha = (1+(1-\alpha)r)^n.
\eeam
Then, there exists an $\bar{r}\in\bbr_+$ with
\beam\label{26.6.2018.2}
(1+R)(1+\bar{r})(1-\alpha) + \alpha < (1+(1-\alpha)r)^{n+1},
\eeam
but 
\beam\label{26.6.2018.3}
(1+R)(1+\bar{r})^2(1-\alpha) + \alpha > (1+(1-\alpha)r)^{n+2}.
\eeam
\item[(iii)] Let $r_1\in\bbr_+$ and $m_1,m_2\in\bbn$ with $m_1\le m_2$. We have the implication
\beao
(1+r_1)^{m_1}(1-\alpha) + \alpha \ge (1+(1-\alpha)r)^{m_1} \implies  
(1+r_1)^{m_2}(1-\alpha) + \alpha \ge (1+(1-\alpha)r)^{m_2}.
\eeao
\end{itemize}
\end{Notiz}
\begin{proof}
Ad (i): From (\ref{27.6.2018.1}), it follows that $\bar{r}>(1-\alpha)r$. 
Since, the difference of the LHSs of (\ref{27.6.2018.2}) and (\ref{27.6.2018.1}) reads $(R'-R)(1+\bar{r})(1-\alpha)$, and the difference of the RHSs is given by $(R'-R)(1-\alpha)(1+(1-\alpha)r)$, one arrives at (\ref{27.6.2018.2}).
 
Ad (ii): Let $R$ be given by (\ref{26.6.2018.1}) and define $\bar{r}$ through {\em equality} in (\ref{26.6.2018.2}). This implies that 
(\ref{27.6.2018.1}) is satisfied. Applied to $R'$ given by $1+R'=(1+R)(1+\bar{r})$, assertion~(i) yields (\ref{26.6.2018.3}).
The assertion follows by choosing $\bar{r}$ slightly smaller such that (\ref{26.6.2018.3}) still holds.

Ad (iii): Let $m_0:=\inf\left\{m\in\bbn\ |\ (1+r_1)^m(1-\alpha) + \alpha \ge (1+(1-\alpha)r)^m\right\}$. The infimum is finite iff $r_1>(1-\alpha)r$. We can assume this since otherwise there is nothing to show. One has
\beao
(1+r_1)^{m_0-1}(1+r_1)(1-\alpha) + \alpha & = & (1+r_1)^{m_0}(1-\alpha) + \alpha\\ 
& \ge & (1+(1-\alpha)r)^{m_0}\\
 & \ge & \left[(1+r_1)^{m_0-1}(1-\alpha) + \alpha\right](1+(1-\alpha)r),
\eeao
which implies that for $R=(1+r_1)^{m_0-1}-1$ the corresponding $\bar{r}$ from (\ref{27.6.2018.1}) satisfies 
\beam\label{20.9.2018.1}
\bar{r}\le r_1.
\eeam
Now, we are in the position to show by induction that  
\beao
(1+r_1)^{m_0+k}(1-\alpha) + \alpha \ge (1+(1-\alpha)r)^{m_0+k},\quad \forall k\in\bbn_0,
\eeao
which completes the proof. Assume that the assertion holds for some $k\in\bbn_0$. We derive that
\beao
(1+r_1)^{m_0+k+1}(1-\alpha) + \alpha & = & (1+r_1)^{m_0+k}(1+r_1)(1-\alpha) 
+ \alpha\\
& \ge & (1+r_1)^{m_0+k}(1+\bar{r})(1-\alpha) + \alpha\\
& > & \left[(1+r_1)^{m_0+k}(1-\alpha) + \alpha\right](1+(1-\alpha)r)\\
& \ge & (1+(1-\alpha)r)^{m_0+k+1}.
\eeao
Here, the first inequality holds by (\ref{20.9.2018.1}), the second by part~(i) applied to
$R=(1+r_1)^{m_0-1}-1$ and $R'=(1+r_1)^{m_0+k}-1$, and the third by the induction hypothesis.
\end{proof}
Note~\ref{26.6.2018.5} allows for the following economic interpretation. Consider a stock position
whose ratio between the unrealized book profit and the pre-tax value reads
$R/(1+R)$. Then, the number~$\bar{r}$ in 
(\ref{27.6.2018.1}) is the minimal deterministic return in the next period
such that it is worthwhile to hold the
stock for one more period instead of liquidating it immediately.
For $R'>R$, the ratio~$R'/(1+R')$ is larger than $R/(1+R)$, and the above break-even point for the stock return of the next period decreases.\\

The first example of this section is about a boundary~$\kappa<r$ larger than $\kappa_{t,T}$ from (\ref{30.10.2017.4}) s.t $P\left((S_t-S_{t-1})/S_{t-1} < \kappa\ |\ \mathcal{F}_{t-1}\right)>0$\ $P$-a.s., but (RLNA) in $t$ does not hold. This means that the risk 
of a loss larger than $-\kappa S_{t-1}$ does not make it impossible that 
a long stock position in period~$t$ 
triggers an arbitrage. 
 
\begin{thmBSP}[On the maximality of $\kappa_{t,T}$]\label{9.11.2017.2}
Let $t,T\in \bbn$  with $2\le t\le T-1$ and $\Omega=\{\omega_1,\omega_2\}$. The outcome~$\omega$ is revealed at time~$t$, i.e., $\mathcal{F}_u=\{\emptyset,\Omega\}$ for $u\le t-1$ and $\mathcal{F}_u=2^\Omega$ for $u\ge t$. 
Let $\kappa$ be a boundary satisfying
\beam\label{3.12.2017.6}
\kappa > \frac{(1-\alpha) \left[ (1+(1-\alpha)r)^T - \alpha\right]}{\left[ (1+(1-\alpha)r)^{t-1} - \alpha\right]\left[ (1+(1-\alpha)r)^{T-t} - \alpha\right]} -1.
\eeam 
The RHS of (\ref{3.12.2017.6}) tends to $\kappa_{t,T}$ for $t,T\to\infty$ and $T-t$ fixed.
In the following, we construct  a stochastic stock return with $P((S_t - S_{t-1})/S_{t-1}< \kappa\ |\ \mathcal{F}_{t-1}) >0$, but $S$ does not satisfy (RLNA) in $t$. 

We assume that 
$S_t(\omega_1)=S_{t-1}(\omega_1)(1+\un{r}_2)$ and
$S_t(\omega_2)=S_{t-1}(\omega_2)(1+\ov{r}_2)$ with parameters $\un{r}_2<\ov{r}_2$ that still have to be specified.
To show that $S$ does not satisfy (RLNA) in $t$, one has to find a process $\wh{S}$ 
which satisfies (NA) in the model from Definition~\ref{17.8.2017.3} with the time domain $\{0,\ldots,t-1,t+1,\ldots,T\}$ such that the corresponding $\wt{S}$ from (\ref{17.8.2017.1}) allows for an arbitrage. We consider 
\beao
\wh{S}_u(\omega) 
:= \left\{ \begin{array}{l@{\ :\ }l}    
                         (1+r_1)^u   & u\le t-1\\    
                          (1+r_1)^{t-1}(1+\ov{r}_3)^{u-t} & u\ge t+1,\ \omega=\omega_1\\
                          (1+r_1)^{t-1} & u\ge t+1,\ \omega=\omega_2
                         \end{array}\right.     
\eeao
where the parameters $r_1,\ov{r}_3 >0$  are also not yet specified. The associated process $\wt{S}$ reads 
\beao
\wt{S}_u(\omega) 
= \left\{ \begin{array}{l@{\ :\ }l}    
                         (1+r_1)^u   & u\le t-1\\    
                          (1+r_1)^{t-1}(1+\un{r}_2)(1+\ov{r}_3)^{u-t} & u\ge t,\ \omega=\omega_1\\
                          (1+r_1)^{t-1}(1+\ov{r}_2) & u\ge t,\ \omega=\omega_2
                         \end{array}\right.     
\eeao
In the following, we state four conditions 
from which we then show that they can be satisfied simultaneously
by a suitable choice of the parameters $r_1,\un{r}_2,\ov{r}_2$, and $\ov{r}_3$.
Later on, we show that $\wh{S}$ satisfies (NA) and $\wt{S}$ allows for an arbitrage under these conditions that read: 
\beam\label{3.12.2017.1}
(1+r_1)^{t-1}(1-\alpha) + \alpha < (1+(1-\alpha)r)^{t-1}\qquad
\mbox{no-arbitrage up to $t-1$,}
\eeam
\beam\label{3.12.2017.2}
(1+\ov{r}_3)^{T-t}(1-\alpha) + \alpha < (1+(1-\alpha)r)^{T-t}\qquad\mbox{no-arbitrage after $\omega$ is revealed,}
\eeam
\beam\label{3.12.2017.4}
& & (1+r_1)^{t-1}(1+\un{r}_2)(1+\ov{r}_3)^{T-t}(1-\alpha) + \alpha\nonumber\\
& &  > (1+(1-\alpha)r)^T\qquad\mbox{profit by buy-at-$0$-and-sell-at-$T$ if $\omega_1$ occurs,}
\eeam
and
\beam\label{3.12.2017.5}
& & (1+r_1)^{t-1}(1+\ov{r}_2)(1-\alpha) + \alpha\nonumber\\
& &  > (1+(1-\alpha)r)^t\qquad\mbox{profit by buy-at-$0$-and-sell-at-$t$ if $\omega_2$ occurs}.
\eeam
By (\ref{3.12.2017.6}), there exists an $\un{r}_2<\kappa$ with
\beam\label{20.9.2018.2}
(1+\un{r}_2)\left[ (1+(1-\alpha)r)^{t-1} - \alpha\right]\left[ (1+(1-\alpha)r)^{T-t} - \alpha\right]  
> (1-\alpha) \left[ (1+(1-\alpha)r)^T - \alpha\right].
\eeam
Fixing such an $\un{r}_2$, one can find $r_1$ and $\ov{r}_3$ such that
(\ref{3.12.2017.1}), (\ref{3.12.2017.2}), and (\ref{3.12.2017.4}) hold simultaneously. Indeed, if $r_1$ and $\ov{r}_3$ were defined through equality in
(\ref{3.12.2017.1}) and (\ref{3.12.2017.2}), respectively, then (\ref{3.12.2017.4})
would be just a reformulation of (\ref{20.9.2018.2}). Now, one chooses $r_1$ and $\ov{r}_3$ slightly smaller s.t. (\ref{3.12.2017.4}) still holds. 
 
After these parameters are already specified, by $r_1>-1$,
one can choose $\ov{r}_2$ large enough such that (\ref{3.12.2017.5}) holds. 
Putting, together inequalities 
(\ref{3.12.2017.1})-(\ref{3.12.2017.5}) can be satisfied simultaneously.\\

Let us show that $\wh{S}$ satisfies (NA). We consider the gains defined in
(\ref{22.12.2017.1}). By (\ref{3.12.2017.1}) and Note~\ref{26.6.2018.5}(iii), one has 
\beam\label{21.9.2018.1}
(1+r_1)^n(1-\alpha) + \alpha < (1+(1-\alpha)r)^n,\quad \forall n\in\bbn\ \mbox{with}\ n\le t-1.
\eeam
Since the stock's return vanishes on $\{\omega_2\}$ after $t-1$, we get
\beam\label{1.7.2018.3}
\wh{X}_{s,u}(\omega_2) & = & \left[(1+r_1)^{u\wedge(t-1)-s\wedge(t-1)}(1-\alpha)+\alpha
- (1+(1-\alpha)r)^{u-s-1_{(s\le t-1<u)}}\right]\nonumber\\
& & \times \ (1+(1-\alpha)r)^{T-u-1_{(u\le t-1)}} <0,\quad s,u\not=t,\ s<u,
\eeam
where the inequality follows from (\ref{21.9.2018.1}).
On the other hand, by (\ref{3.12.2017.2})
and again Note~\ref{26.6.2018.5}(iii), one has that 
\beam\label{1.7.2018.4}
\wh{X}_{s,u}(\omega_1) & = &
\left[(1+\ov{r}_3)^{u-s\vee t}(1-\alpha)+\alpha
- (1+(1-\alpha)r)^{u-s\vee t}\right]\nonumber\\
& & \times\  (1+(1-\alpha)r)^{T-u}<0,\quad t-1\le s<u.
\eeam
Now, let $N$ be some arbitrary strategy in the market~$\{0,\ldots,t-1,t+1,\ldots,T\}$ with $\wh{V}(N)\ge 0$, cf. (\ref{27.11.2017.2}).
From (\ref{1.7.2018.3}) it follows that $N_{s,u-1}(\omega_2)-N_{s,u}(\omega_2)=0$ for all $s<u$ and thus
$N_{s,s}(\omega_2) = \sum_{u=s+1}^T(N_{s,u-1}(\omega_2)-N_{s,u}(\omega_2))=0$.
Since $\mathcal{F}_{t-2}$ is trivial, this implies that $N_{s,s}=0$ for all $s\le t-2$. In addition, $N_{s,s}=0$ for all $s\ge t-1$ by 
(\ref{1.7.2018.4})/(\ref{1.7.2018.3}). Thus, $\wh{V}(N)$ from (\ref{27.11.2017.2}) vanishes and $\wh{S}$ satisfies (NA).

On the other hand, in the model with price process~$\wt{S}$, the strategy 
\beao
N_{s,u}(\omega) = 1_{( \omega=\omega_1,\ s=0,\ u\le T-1)} + 1_{( \omega=\omega_2,\ s=0,\ u\le t-1)},
\eeao
leading to 
\beao
\wt{V}(N)= \wt{X}_{0,T}1_{\{\omega_1\}} + \wt{X}_{0,t}1_{\{\omega_2\}},
\eeao
 is an arbitrage. Namely, one has $\wt{X}_{0,T}(\omega_1)>0$ and $\wt{X}_{0,t}(\omega_2)>0$ by (\ref{3.12.2017.4}) and (\ref{3.12.2017.5}), respectively. Putting together, $S$ does not satisfy (RLNA) in $t$.
\end{thmBSP} 

Example~\ref{9.11.2017.2} is based on two features: first, $\un{r}_2$, the bad return in period~$t$, is compensated by the deferment of taxes across period~$t$; 
and second, there exists the chance of a good return~$\ov{r}_2$ that makes a purchase before~$t-1$ profitable even if the stock has to be liquidated after period~$t$. 
This return is missing in the model with price process~$\wh{S}$ on the event that
the stock has to be liquidated after period~$t$.
Put differently, the stochastic return in period~$t$ can be used to hedge against the bad performance of the stock afterwards. 

On the other hand, if the period~$t$ return {\em never} exceeded $(1-\alpha)r$, the addition of period~$t$ would not provide any advantage.
Especially one cannot construct  deterministic examples for the same boundaries as in (\ref{3.12.2017.6}).

\begin{thmBSP}[Two long positions in the same stock that hedge each other]\label{9.11.2017.4}
Let $T=3$, $\Omega=\{\omega_1,\omega_2\}$, $\mathcal{F}_0=\mathcal{F}_1=\{\emptyset,\Omega\}$, and $\mathcal{F}_2=\mathcal{F}_3=2^\Omega$, i.e., 
$\omega$ is revealed at time~$2$. 
We consider the following stock price with parameters $\bar{r},\eps_1,\eps_2>0$ that
still have to be specified: $S_0=1$, $S_1=1+\bar{r}$, 
\beao
S_2(\omega_1) = (1+\bar{r})(1+r-\eps_1),
\quad
S_3(\omega_1) = (1+\bar{r})(1+r-\eps_1)(1+\bar{r})
\eeao
and
\beao
S_2(\omega_2) = (1+\bar{r})(1+r+\eps_2),\quad S_3(\omega_2)=0.
\eeao

The return~$\bar{r}\in (0,r)$ has to satisfy 
\beam\label{12.11.2017.1}
(1+\bar{r})(1+r)(1-\alpha) + \alpha < (1+(1-\alpha)r)^2,
\eeam
but
\beam\label{12.11.2017.2}
(1+\bar{r})^2(1+r)(1-\alpha) + \alpha > (1+(1-\alpha)r)^3.
\eeam
By Note~\ref{26.6.2018.5}(ii), the two conditions on $\bar{r}$ can be satisfied simultaneously. 
From now on, we fix such an $\bar{r}$. Note that the gain of the short-term investment in period~2 reads
\beam\label{4.7.2018.1}
X_{1,2}(\omega) 
= \left\{ \begin{array}{l@{\ :\ }l}    
                         -(1-\alpha)(1+(1-\alpha)r)(1+\bar{r})\eps_1 & \quad\mbox{for\ }\omega=\omega_1\\
                         (1-\alpha)(1+(1-\alpha)r)(1+\bar{r})\eps_2 & \quad\mbox{for\ }\omega=\omega_2 
                         \end{array}\right.     
\eeam
We proceed by choosing $\eps_1,\eps_2>0$ both ``small'' such that
\beam\label{9.11.2017.1}
X_{1,2}(\omega_1)X_{0,2}(\omega_2) - X_{1,2}(\omega_2) X_{0,3}(\omega_1) = 0,
\eeam
\beam\label{12.2.2018.1}
X_{0,3}(\omega_1) = (1+\bar{r})(1+r-\eps_1)(1+\bar{r})(1-\alpha) + \alpha - (1+(1-\alpha)r)^3 > 0,
\eeam
and
\beam\label{12.2.2018.2}
X_{0,2}(\omega_2) = \left[(1+\bar{r})(1+r+\eps_2)(1-\alpha) + \alpha\right] 
(1+(1-\alpha)r) - (1+(1-\alpha)r)^3 < 0.
\eeam
To achieve this, define $F(\eps_1,\eps_2)$ as the LHS of (\ref{9.11.2017.1}). One has 
\beao
\partial_2 F(0,0) = -(1-\alpha)(1+\bar{r})(1+(1-\alpha)r)
\left[(1+\bar{r})^2(1+r)(1-\alpha) + \alpha - (1+(1-\alpha)r)^3\right] < 0.
\eeao
Thus, the implicit function theorem yields the existence of a 
function~$f$ defined in a neighborhood of $0$ with $F(\eps_1,f(\eps_1))=0$
for all $\eps_1>0$ small enough.  
We choose $\eps_1>0$ small enough and $\eps_2=f(\eps_1)$. By (\ref{12.11.2017.1})
and (\ref{12.11.2017.2}), the conditions (\ref{12.2.2018.2}), (\ref{12.2.2018.1}), and (\ref{9.11.2017.1}) can be satisfied simultaneously.\\

We have constructed a model with two states and two self-financing and opposite investment opportunities: a long-term investment that
buys the stock at time~$0$ and sells it at time~$3$
if $\omega_1$ occurs or at time~$2$ if $\omega_2$ occurs and a short-term investment
that buys the stock at time~$1$ and sells it at time~2. All other investments lead to sure losses. Equation~(\ref{9.11.2017.1}) is a no-arbitrage condition ensuring that one of the two investment opportunities is redundant. Consider the strategy
\beam\label{29.12.2017.1}
N_{0,0}=N_{0,1}=1,\ N_{0,2}=1_{\{\omega_1\}},\quad N_{1,1}=-X_{0,3}(\omega_1)/X_{1,2}(\omega_1),\ N_{1,2}=0,\ \mbox{and}\ N_{2,2}=0,
\eeam
in which $X_{0,3}(\omega_1)>0$ and $-X_{1,2}(\omega_1)>0$ by (\ref{12.2.2018.1}) and (\ref{4.7.2018.1}), respectively. The liquidation value reads 
\beao
V(N) = X_{0,3}1_{\{\omega_1\}} + X_{0,2}1_{\{\omega_2\}} + (- X_{0,3}(\omega_1)/X_{1,2}(\omega_1))X_{1,2},
\eeao 
which vanishes by (\ref{9.11.2017.1}). When buying one stock at time zero, $- X_{0,3}(\omega_1)/X_{1,2}(\omega_1)$ stocks at time one and following the liquidation rules from above, the after-tax gains of the two long stock positions cancel each other out, and the total gain after financing costs 
disappears for sure.\\

Finally, we give a detailed proof that the model satisfies (NA). For stocks purchased at time~$0$, we have that
\beam\label{1.7.2018.01}
\max_{u=1,2,3}X_{0,u}(\omega_1) = X_{0,3}(\omega_1)>0\quad\mbox{and}\quad \max_{u=1,2,3}X_{0,u}(\omega_2) = X_{0,2}(\omega_2)<0.
\eeam
Indeed, the former holds since $X_{0,2}(\omega_1)<0$ by (\ref{12.11.2017.1}) and 
$X_{0,3}(\omega_1)>0$ by
(\ref{12.2.2018.1}). The latter holds by 
$X_{0,3}(\omega_2)<X_{0,1}(\omega_2)<X_{0,2}(\omega_2)$ and (\ref{12.2.2018.2}).
For stocks purchased at time~$1$, we get
\beam\label{1.7.2018.02}
\max_{u=2,3}X_{1,u}(\omega_1) = X_{1,2}(\omega_1)<0\quad\mbox{and}\quad \max_{u=2,3}X_{1,u}(\omega_2) = X_{1,2}(\omega_2)>0.
\eeam
Here, $X_{1,2}(\omega_1)<0$ holds by (\ref{4.7.2018.1}) and $X_{1,2}(\omega_1)>X_{1,3}(\omega_1)$ follows from (\ref{12.11.2017.1}), $-\eps_1<0$, and Note~\ref{26.6.2018.5}(i). On the other hand, $X_{1,2}(\omega_2)>0$ by (\ref{4.7.2018.1}) and $X_{1,3}(\omega_2)<0$.\\

Now, let $N$ be some arbitrary strategy. By (\ref{1.7.2018.01}) and $\sum_{u=1}^3(N_{0,u-1}-N_{0,u}) = N_{0,0}$, we have the estimate
\beao
\sum_{u=1}^3 (N_{0,u-1}-N_{0,u})X_{0,u}
\le   N_{0,0} \left( X_{0,3}1_{\{\omega_1\}} + X_{0,2}1_{\{\omega_2\}}\right).
\eeao
With (\ref{1.7.2018.02}), we have an analogue estimate for the gains of purchases at time~$1$ and since $X_{2,3}<0$, we arrive at
\beao
V(N) = \sum_{s=0}^2 \sum_{u=s+1}^3 (N_{s,u-1}-N_{s,u})X_{s,u}
\le   
N_{0,0}\left( X_{0,3}1_{\{\omega_1\}} + X_{0,2}1_{\{\omega_2\}}\right) + N_{1,1}X_{1,2}.
\eeao
One has
\beam\label{30.6.2018.1}
V(N) & \le  & N_{0,0}\left( X_{0,3}1_{\{\omega_1\}} + X_{0,2}1_{\{\omega_2\}}\right) + N_{1,1}X_{1,2}\nonumber\\
& = & N_{0,0} X_{1,2}\frac{X_{0,3}(\omega_1)}{X_{1,2}(\omega_1)}  + N_{1,1}X_{1,2}\nonumber\\
& = & \left[N_{0,0}\frac{X_{0,3}(\omega_1)}{X_{1,2}(\omega_1)} + N_{1,1}\right]X_{1,2},
\eeam
where the first equality holds by (\ref{9.11.2017.1}). Now, assume that
$V(N)\ge 0$. Since $X_{1,2}$ can take both a positive and a negative value, (\ref{30.6.2018.1}) implies that the deterministic prefactor vanishes, i.e., $N_{0,0}X_{0,3}(\omega_1)/X_{1,2}(\omega_1) + N_{1,1}=0$ and thus $V(N)=0$.
Consequently, the market satisfies (NA).
\end{thmBSP}

We have constructed an arbitrage-free model with two long positions in the same stock that hedge each other. The key feature of Example~\ref{9.11.2017.4} is that a bad return in period~$2$ is followed by a good return in period~$3$ and vice versa. For a short-term investor, it is only worthwhile to speculate in the return of period $2$ (thus, the advance information about a good return in period~$3$ cannot be used for an arbitrage). By contrast, for a long-term investor who already buys the stock at time zero and accepts lower returns since she profits from the deferment of taxes, the return of period~$3$ has a stronger impact. 
Thus, a long-term investor believing in this stochastic model hopes for the bad return in period~$2$.

Such an example cannot exist in an arbitrage-free model without taxes because the difference between long- and short-term investments disappears. This is shown in the following remark.
\begin{Bemerkung}\label{7.1.2018.1}
Let $(\un{S},\ov{S})$ be the discounted discrete time bid-ask-price process of a risky stock in a model without capital gains taxes. If the market enriched by a riskless bank account is arbitrage-free, there exists a $Q\sim P$ and a 
$Q$-martingale $S$ with $\un{S}\le S\le \ov{S}$ (for the result on general probability spaces, see Corollary~2.9 in 
Grigoriev~\cite{grigoriev.2005}). Then, for any dynamic strategy with vanishing liquidation value, the gain process~$\varphi\mal S$ in the discounted shadow price also vanishes, where $\varphi$ denotes the predictable number of risky stocks. Indeed, since trading at price process~$S$ is at least as favorable as at the bid-ask prices, one has $Q(\varphi\mal S_T\ge 0)=1$. On the other hand, a discrete time local martingale with nonnegative terminal value
is a true martingale. Together with $\varphi\mal S_0=0$, one arrives at
$Q(\varphi\mal S_t=0,\ t=0,\ldots,T)=1$. However, this implies that each individual share bought and sold by the dynamic strategy between $0$ and $T$ makes zero profit for sure (after financing costs), and thus an effect as in Example~\ref{9.11.2017.4} cannot occur.
\end{Bemerkung}

An extension of the construction 
in Example~\ref{9.11.2017.4} is also useful in the following example. 
Here, we want to show that on an infinite probability space, no-arbitrage alone does not imply the existence of a separating measure since the set of attainable terminal wealths does not need to be closed 
regarding the convergence in probability.

\begin{thmBSP}[No-arbitrage $\nRightarrow$ $\exists$ equivalent separating  probability measure]\label{3.1.2018.1}
Let $T=4$ and $\Omega=\{\omega_{n,1}\ |\ n\in\bbn\}\cup  \{\omega_{n,2}\ |\ n\in\bbn\}$, 
$\mathcal{F}_0 = \mathcal{F}_1 = \{\emptyset,\Omega\}$, 
$\mathcal{F}_2 = \sigma(\{\{\omega_{n,1}, \omega_{n,2}\}\ |\ n\in\bbn\})$, and $\mathcal{F}_3=\mathcal{F}_4=2^\Omega$. This means that $n$ is already revealed at time~$2$ and full information at time~$3$. The stock price depends on parameters that still have to be specified and reads $S_0=1$, $S_1=1+r_1$, $S_2=(1+r_1)(1+r)$,
\beao
S_3(\omega_{n,1}) = (1+r_1)(1+r)(1+r_2-\eps_{n,1}),
\quad
S_4(\omega_{n,1}) = (1+r_1)(1+r)(1+r_2-\eps_{n,1})(1+r_1)
\eeao
and
\beao
S_3(\omega_{n,2}) = (1+r_1)(1+r)(1+r_2+\eps_{n,2}),\quad S_4(\omega_{n,2})=0,
\eeao
where $r_2\in (0,r)$ is given by 
\beao
(1+r)(1+r_2)(1-\alpha) + \alpha = (1+(1-\alpha)r)^2.
\eeao 
In addition, we fix an 
$r_1>0$ satisfying 
\beam\label{10.11.2017.1}
 (1+r_1)(1+r)(1+r_2)(1-\alpha) + \alpha < (1+(1-\alpha)r)^3,
\eeam
but
\beam\label{10.11.2017.2}
(1+r_1)(1+r)(1+r_2)(1+r_1)(1-\alpha) + \alpha > (1+(1-\alpha)r)^4.
\eeam
By Note~\ref{26.6.2018.5}(ii), such an $r_1$ exists. Observe that the difference
between the LHS and the RHS of (\ref{10.11.2017.2}) corresponds to $X_{0,4}(\omega_{n,1})$  if $\eps_{n,1}$ is ignored. The same holds for 
(\ref{10.11.2017.1}) and $X_{0,3}(\omega_{n,i})/(1+(1-\alpha)r)$ if $\eps_{n,i}$,\ $i=1,2$, is ignored.

Note~\ref{26.6.2018.5}(i) applied to $R=r$ (which means that $\bar{r}=r_2$) and $R'=(1+r)(1+r_2)-1$
yields that
$(1+r)(1+r_2)^2(1-\alpha) + \alpha > (1+(1-\alpha)r)^3$ and thus by (\ref{10.11.2017.1})
\beam\label{r_1r_2}
r_1<r_2.
\eeam
We proceed by specifying the sequence~$(\eps_{n,2})_{n\in\bbn}\subset\bbr_+\setminus\{0\}$. Let $\eps_2>0$ be small enough such that $(1+r_1)(1+r)(1+r_2+\eps_2)(1-\alpha) + \alpha < (1+(1-\alpha)r)^3$ and $r_2 + \eps_2 <r$, which exists by (\ref{10.11.2017.1}) and $r_2<r$. Then, define 
$\eps_{n,2}:= (1/n)\wedge \eps_2$ for all $n\in\bbn$. This choice ensures that
\beam\label{28.6.2018.01}
\eps_{n,2}\downarrow 0,\ n\to\infty,\quad \sup_{n\in\bbn}X_{0,3}(\omega_{n,2})<0,\quad X_{1,3}(\omega_{n,2})>0,\quad X_{2,3}(\omega_{n,2})<0,\ \forall n\in\bbn.
\eeam
Note that
\beao
X_{1,3}(\omega) 
= \left\{ \begin{array}{l@{\ :\ }l}    
                         -(1+r_1)(1+r)(1-\alpha)(1+(1-\alpha)r)\eps_{n,1} & \quad\mbox{for\ }\omega=\omega_{n,1}\\
                         (1+r_1)(1+r)(1-\alpha)(1+(1-\alpha)r)\eps_{n,2} & \quad\mbox{for\ }\omega=\omega_{n,2} 
                         \end{array}\right.     
\eeao
To complete the construction of the stock price process, it remains to specify the sequence~$(\eps_{n,1})_{n\in\bbn}\subset\bbr_+\setminus\{0\}$.
The goal is that the conditions
\beam\label{28.6.2018.3}
X_{0,4}(\omega_{n,1})>0,\quad X_{1,3}(\omega_{n,1})<0,\ \forall n\in\bbn,
\eeam
and
\beam\label{29.12.2017.5}
X_{1,3}(\omega_{n,2}) X_{0,4}(\omega_{n,1})
-2 X_{1,3}(\omega_{n,1})X_{0,3}(\omega_{n,2}) \ge 0,\quad \forall n\in\bbn,
\eeam
are satisfied. To achieve this, we firstly choose an $\eps_1>0$ satisfying $(1+r_1)(1+r)(1+r_2-\eps_1)(1+r_1)(1-\alpha) + \alpha > (1+(1-\alpha)r)^4$, which exists by (\ref{10.11.2017.2}). Then, we consider the functions $f_n$ given by
\beam\label{28.6.2018.2}
f_n(\eps) & := & X_{1,3}(\omega_{n,2})\left[(1+r_1)(1+r)(1+r_2-\eps)(1+r_1)(1-\alpha) + \alpha -(1+(1-\alpha)r)^4\right]\nonumber\\
& & +2 X_{0,3}(\omega_{n,2})(1+r_1)(1+r)(1-\alpha)(1+(1-\alpha)r)\eps,\qquad n\in\bbn.
\eeam
$f_n(\eps)$ is the LHS of (\ref{29.12.2017.5}) if the number~$\eps_{n,1}$ entering in $X_{0,4}(\omega_{n,1})$ and $X_{1,3}(\omega_{n,1})$ is replaced by the variable~$\eps$. For every $n\in\bbn$, the factor after $X_{1,3}(\omega_{n,2})>0$  in (\ref{28.6.2018.2}) converges to a positive number by (\ref{10.11.2017.2}) and the factor after $X_{0,3}(\omega_{n,2})$ to zero for $\eps\downarrow 0$. Thus, one has $f_n(\eps)>0$ for $\eps>0$ small enough. Consequently,
\beao
\eps_{n,1} := \sup\{1/k\ |\ k\in\bbn,\  1/k < \eps_1,\ \mbox{and}\ f_n(1/k) > 0\}
\eeao
is positive and the condition (\ref{29.12.2017.5}) is satisfied.
By the choice of $\eps_1$ and $\eps_{n,1}\le\eps_1$, one has $X_{0,4}(\omega_{n,1})>0$ for all $n\in\bbn$. In addition, one has $X_{1,3}(\omega_{n,1})<0$ for all $n\in\bbn$ and thus (\ref{28.6.2018.3}).\\

Let us explain the main idea of the construction. As in Example~\ref{9.11.2017.4},
there is a long- and a short-term investment in the stock 
that are in the opposite direction. 
For the long-term investment, the stock is already purchased at time~$0$.
The short-term investment is between $1$ and $3$, but the return between $1$ and $2$ equals the riskless interest rate~$r$. It simply provides the opportunity to defer taxes beyond time~$2$ if the stock is still held in the portfolio, but without a loss if one decides against it at time~$2$. 
The long-term investor hopes for the event~$\{\omega_{n,1}\ |\ n\in\bbn\}$ and the short-term investor for $\{\omega_{n,2}\ |\ n\in\bbn\}$. In contrast to Example~\ref{9.11.2017.4}, the investments do not only cancel each other out, but one can achieve a systematic profit if one holds a suitable ratio.
The problem is that the profit and the loss of the short-term investment 
disappear with $n\to\infty$. Thus, depending on $n$, one needs to buy more and more stocks at time $1$, but $n$ is not revealed before time~$2$. On the other hand, by buying the stock at time~2 and doing without the benefits from the deferment of gains accrued in period~2, one makes a loss for sure.
Without the knowledge of $n$, stocks can be bought ahead at time~$1$ and -- if they are not needed -- sold without any loss at time~$2$. This procedure generates  an approximate arbitrage by buying ahead more and more stocks at time $1$ and selling the stocks which are not needed at time~$2$. The remaining risk that there are not enough short-term stocks in the portfolio at time~$2$ disappears in the limit.\\

{\em Step 1:} Let us first show that the model satisfies (NA). For stock purchases at time~$0$, we have that 
\beam\label{27.6.2018.3}
\max_{u=1,2,3,4}X_{0,u}(\omega_{n,1})=X_{0,4}(\omega_{n,1}),\quad
\max_{u=1,2,3,4}X_{0,u}(\omega_{n,2})=X_{0,3}(\omega_{n,2}),\qquad
 \forall n\in\bbn. 
\eeam
Indeed, $X_{0,4}(\omega_{n,1})>0$ holds by (\ref{28.6.2018.3}), but 
$X_{0,1}<X_{0,2}<0$ by (\ref{r_1r_2}) and $X_{0,3}(\omega_{n,1})<0$ by
(\ref{10.11.2017.1}). For $\omega_{n,2}$, we have that
$X_{0,4}(\omega_{n,2}) < X_{0,1} < X_{0,2}<X_{0,3}(\omega_{n,2})$.
Here, the last inequality follows from $\eps_{n,2}>0$ and Note~\ref{26.6.2018.5}(i) applied to $R=r$ and $R'=(1+r_1)(1+r)-1$.

For stock purchases at time~$1$, we get 
\beam\label{27.6.2018.4}
X_{1,2}=0,\quad X_{1,3} > X_{1,4},
\eeam
in which $X_{1,3}(\omega_{n,1})>X_{1,4}(\omega_{n,1})$ follows from (\ref{10.11.2017.1}), $-\eps_{n,1}<0$, and Note~\ref{26.6.2018.5}(i).
Later purchases lead to sure losses, i.e.,
\beam\label{27.6.2018.5}
 X_{2,3}<0,\quad X_{2,4}<0,\quad\mbox{and}\quad  X_{3,4}<0,
 \eeam
 where $X_{2,3}(\omega_{n,2})<0$ is ensured by (\ref{28.6.2018.01}).
 
Now, let $N$ be some arbitrary strategy. By (\ref{27.6.2018.3}) and $\sum_{u=1}^4(N_{0,u-1}-N_{0,u})=N_{0,0}$, we have the estimate
\beao
\sum_{u=1}^4(N_{0,u-1}-N_{0,u})X_{0,u} 
\le 
N_{0,0}\left(X_{0,4}1_A + X_{0,3}1_{\Omega\setminus A}\right),\quad
\mbox{where\ }A:=\{\omega_{n,1}\ |\ n\in\bbn\}.
\eeao 
By (\ref{27.6.2018.4}), we get the estimate
\beao
\sum_{u=2}^4(N_{1,u-1}-N_{1,u})X_{1,u} 
= \sum_{u=3}^4(N_{1,u-1}-N_{1,u})X_{1,u}
\le 
N_{1,2} X_{1,3}. 
\eeao 
Finally, by (\ref{27.6.2018.5}), one has
$\sum_{u=3}^4(N_{2,u-1}-N_{2,u})X_{2,u}\le 0$ and $N_{3,3}X_{3,4}\le 0$. Putting together, 
\beam\label{13.11.2017.1}
V(N) =\sum_{s=0}^3\sum_{u=s+1}^4 (N_{s,u-1}-N_{s,u})X_{s,u} \le N_{0,0}\left(X_{0,4}1_A + X_{0,3}1_{\Omega\setminus A}\right) + N_{1,2} X_{1,3}.
\eeam
Now, assume that $V(N)\ge 0$ and let us show that this implies $V(N)=0$. Since $N_{1,2}$ is bounded from above by the real number $N_{1,1}$,
$X_{1,3}(\omega_{n,2})=(1+r_1)(1+r)(1-\alpha)(1+(1-\alpha)r)\eps_{n,2}\to 0$ for $n\to\infty$ by (\ref{28.6.2018.01}), and $\sup_{n\in\bbn}X_{0,3}(\omega_{n,2})<0$, we conclude from the non-negativity of the RHS of (\ref{13.11.2017.1}) on $\Omega\setminus A$ that $N_{0,0}=0$. From $X_{1,3}(\omega_{n,1})<0$ for all $n\in\bbn$ and the non-negativity of the RHS of (\ref{13.11.2017.1}) on $A$, it follows that the $\mathcal{F}_2$-measurable nonnegative random variable~$N_{1,2}$ vanishes.
We arrive at $V(N)\le 0$, which means that there cannot be an arbitrage.\\

{\em Step 2:} Consider the sequence of strategies~$(N^m)_{m\in\bbn}$ given by
\beao
& & N^m_{0,0}=N^m_{0,1}=N^m_{0,2}=1,\ N^m_{0,3} = 1_A,\\
& & N^m_{1,1}=m,\ N^m_{1,2}(\omega_{n,i}) 
= \left(-X_{0,4}(\omega_{n,1})/X_{1,3}(\omega_{n,1})\right)\wedge m,\ n\in\bbn,\ i=1,2,\ N^m_{1,3}=0,
\eeao
and $N_{s,u}=0$ for $s\ge 2$. By $X_{1,2}=0$,
(\ref{13.11.2017.1}) holds with equality for $N^m$.
For $\omega=\omega_{n,1}$, the long- and short-term investments cancel each other out if $m\ge |X_{0,4}(\omega_{n,1})/X_{1,3}(\omega_{n,1})|$, and the profit in the long-term investment~$X_{0,4}(\omega_{n,1})>0$ dominates if $m< |X_{0,4}(\omega_{n,1})/X_{1,3}(\omega_{n,1})|$. This means that $V(N^m)1_A = N^m_{0,0} X_{0,4}1_A + N^m_{1,2} X_{1,3} 1_A \ge 0$.

For $\omega=\omega_{n,2}$ and if  $m\ge |X_{0,4}(\omega_{n,1})/X_{1,3}(\omega_{n,1})|$, we get 
\beao
V(N^m)(\omega_{n,2}) & = & X_{0,3}(\omega_{n,2}) - \frac{X_{0,4}(\omega_{n,1})}{X_{1,3}(\omega_{n,1})}X_{1,3}(\omega_{n,2})\\
& \ge & \frac12 \frac{X_{0,4}(\omega_{n,1})}{X_{1,3}(\omega_{n,1})}X_{1,3}(\omega_{n,2}) -\frac{X_{0,4}(\omega_{n,1})}{X_{1,3}(\omega_{n,1})}X_{1,3}(\omega_{n,2})\\
& = & -\frac12 \frac{X_{0,4}(\omega_{n,1})}{X_{1,3}(\omega_{n,1})}X_{1,3}(\omega_{n,2})>0,
\eeao
where the first inequality holds by by (\ref{29.12.2017.5}).
For $\omega=\omega_{n,2}$ and if  $m< |X_{0,4}(\omega_{n,1})/X_{1,3}(\omega_{n,1})|$,
one has by $X_{1,3}(\omega_{n,2})>0$ the trivial estimate
$V(N^m)(\omega_{n,2}) \ge X_{0,3}(\omega_{n,2})$, that does however not ensure non-negativity. Putting together, one obtains that
\beam\label{11.2.2018.1}
V(N^m) & \ge & V(N^m) 1_{\Omega\setminus A}\nonumber\\
& \ge & 
-\frac12 \sum_{n=1}^\infty \frac{X_{0,4}(\omega_{n,1})}{X_{1,3}(\omega_{n,1})} X_{1,3}(\omega_{n,2}) 1_{(m\ge |X_{0,4}(\omega_{n,1})/X_{1,3}(\omega_{n,1})|)}1_{\{\omega_{n,2}\}}\nonumber\\
& & +  \sum_{n=1}^\infty X_{0,3}(\omega_{n,2}) 1_{(m< |X_{0,4}(\omega_{n,1})/X_{1,3}(\omega_{n,1})|)}1_{\{\omega_{n,2}\}}.
\eeam
The RHS of (\ref{11.2.2018.1}) converges pointwise to 
\beam\label{24.11.2018.2}
-\frac12 \sum_{n=1}^\infty \frac{X_{0,4}(\omega_{n,1})}{X_{1,3}(\omega_{n,1})} X_{1,3}(\omega_{n,2}) 1_{\{\omega_{n,2}\}}
\in L^0_+(\Omega,\mathcal{F},P)\setminus\{0\}
\eeam
for $m\to\infty$. This means that there exists an approximate arbitrage.
Putting together, for
$\mathcal{A}:=\{\eta_T\in L^0(\Omega,\mathcal{F},P)\ |\ (\eta,N)\ \mbox{self-financing}\} - L^0_+(\Omega,\mathcal{F},P)$,
one has that 
\beam\label{26.1.2018.1}
\mathcal{A}\cap L^0_+(\Omega,\mathcal{F},P)=\{0\},\quad\mbox{but}\quad
\ov{\mathcal{A}}\cap L^0_+(\Omega,\mathcal{F},P)\supsetneq\{0\},
\eeam
where $\ov{\mathcal{A}}$ denotes the closure of $\mathcal{A}$ regarding the convergence in probability. 

A minimal condition on a separating probability measure~$Q\sim P$
is that
\beam\label{24.11.2018.1}
E_Q(\zeta)\le 0,\quad\forall \zeta\in \mathcal{A}\cap L^\infty(\Omega,\mathcal{F},P).
\eeam
Since the sequence~$(V(N^m))_{m\in\bbn}$ is uniformly bounded from below by 
\beao
\inf_{n\in\bbn}X_{0,3}(\omega_{n,2})
=  \left[(1+r_1)(1+r)(1+r_2)(1-\alpha) + \alpha - (1+(1-\alpha)r)^3\right](1+(1-\alpha)r) > -\infty,
\eeao
it immediately follows from (\ref{24.11.2018.2}) and Fatou's lemma that (\ref{24.11.2018.1}) cannot hold.
\end{thmBSP}

The construction in Example~\ref{3.1.2018.1} can also be used to establish a proportional transaction costs model with 3 assets which satisfies (NA), 
although an equivalent separating probability  measure 
does not exist. This means that Grigoriev's theorem cannot be extended to dimension~3 (see Theorem~1.2 and the discussion in Section~5 of 
\cite{grigoriev.2005}). 

\begin{Beispiel}[Counterexample for an extension of Grigoriev's theorem
to dimension~3]\label{26.1.2018.2}
Let $T=2$, 
$\Omega=\{\omega_{n,1}\ |\ n\in\bbn\}\cup  \{\omega_{n,2}\ |\ n\in\bbn\}$, 
$\mathcal{F}_0 = \{\emptyset,\Omega\}$, 
$\mathcal{F}_1 = \sigma(\{\{\omega_{n,1}, \omega_{n,2}\}\ |\ n\in\bbn\})$, and $\mathcal{F}_2=2^\Omega$. This means that $n$ is already revealed at time~$1$ and full information at time~$2$. Besides a bank account that does not pay interest, there are two risky stocks 
with the bid-ask-price processes~$(\un{S}^1,\ov{S}^1)$ and 
$(\un{S}^2,\ov{S}^2)$. 
The initial prices satisfy $\un{S}^1_0=\ov{S}^1_0=\un{S}^2_0=\ov{S}^2_0=1$. 
At time~$1$, stock~2 can be sold at price~1, i.e., $\un{S}^2_1=1$, and other prices are sufficiently unfavorable to avoid trades, namely $\un{S}^1_1=0$, $\ov{S}^1_1=3$, and $\ov{S}^2_1=3$.
The terminal prices are
\beao
& & \un{S}^1_2(\omega_{n,1})=\ov{S}^1_2(\omega_{n,1})=3,\quad \un{S}^1_2(\omega_{n,2})=\ov{S}^1_2(\omega_{n,2})=0,\\
& & 
\un{S}^2_2(\omega_{n,1})=\ov{S}^2_2(\omega_{n,1})=1-1/n,\quad
\un{S}^2_2(\omega_{n,2})=\ov{S}^2_2(\omega_{n,2}) =1+1/n,\ n\in\bbn.
\eeao   
This means that at times~$0$ and $2$ we have a frictionless market and at time~$1$
one can only sell the stock~$2$ (at the same price as at time~$0$, but after observing ~$n$). 

{\em Step 1:} Let us show (NA). It is sufficient to consider strategies that are static apart from possible sells of stock~2 at time~1. The liquidation values dominate those  of
other strategies, and for zero initial capital they can be written as
\beam\label{30.1.2018.1}
V:= \sum_{n=1}^\infty \left( 2\vp^1_1 -\frac1{n}\left(\vp^2_1 + \Delta \vp^2_2\right)
\right) 1_{\{\omega_{n,1}\}}
+ \sum_{n=1}^\infty \left( -\vp^1_1 +\frac1{n}\left(\vp^2_1 + \Delta \vp^2_2\right)
\right) 1_{\{\omega_{n,2}\}},
 \eeam
where $\vp^1_1$, $\vp^2_1$ are arbitrary real numbers and $\Delta \vp^2_2$ is a non-positive $\mathcal{F}_1$-measurable random variable.
Following the standard notation is discrete time finance, the quantities denote the amounts of stocks in the portfolio. 
Assume that $V\ge 0$. 
By $\Delta \vp^2_2\le 0$ and $1/n\to 0$ for $n\to\infty$, the non-negativity of the second sum in (\ref{30.1.2018.1}) implies $\vp^1_1\le 0$. On the other hand, one has $\vp^1_1 \le \left(\vp^2_1 + \Delta \vp^2_2\right)/n\le 2\vp^1_1$. Putting together, we obtain $\vp^1_1=0$ and thus $\vp^2_1 + \Delta \vp^2_2=0$ as well as $V=0$.
This means that the market satisfies (NA).

{\em Step 2:} On the other hand, an ``approximate arbitrage'' is given by the 
sequence~$(\vp^{1,m}_1,\vp^{2,m}_1,\Delta \vp^{2,m}_2)_{m\in\bbn}$ with 
$\vp^{1,m}_1=1$, $\vp^{2,m}_1=m$, and $\Delta \vp^{2,m}_2=-(m-n)^+$ on $\{\omega_{n,1},\omega_{n,2}\}$.  The corresponding liquidation values read 
$\sum_{n=1}^\infty \left( 2 -\frac{n\wedge m}{n}\right) 1_{\{\omega_{n,1}\}}
+ \sum_{n=1}^\infty \left( -1 +\frac{n\wedge m}{n}\right)1_{\{\omega_{n,2}\}}$,\ $m\in\bbn$, which converges pointwise to $1_{\{\omega_{n,1}\ |\ n\in\bbn\}}$ for $m\to\infty$.
This yields (\ref{26.1.2018.1}). 

By Proposition~3.2.6 of Kabanov and Safarian~\cite{kabanov.safarian.2009}, $\ov{\mathcal{A}}\cap L^0_+(\Omega,\mathcal{F},P)\not=\{0\}$ already implies that a so-called consistent price system cannot 
exist (this means that one need not make use of the fact that the above sequence leads to bounded losses and thus Fatou's lemma is applicable). 
\end{Beispiel}
Example~3.1 of Schachermayer~\cite{schachermayer.2004} demonstrates the same property in a one-period model with $4$ assets. The main idea of our example is the same as in
\cite{schachermayer.2004}. Roughly speaking, there is an asset for which the need is initially unknown and unbounded, and after the need is revealed, it can be sold at its initial purchasing price. Thus, one buys it ahead more and more, and the limit of this strategy does not exist.
On the other hand, the ``currency example'' in \cite{schachermayer.2004} is based on the property that no asset can play the role of a bank account, i.e., of an asset that is involved in every transaction.  Indeed, in {\em one-period} transaction costs models with
a bank account, Grigoriev's theorem holds for arbitrary many assets; namely, these models can be written as frictionless markets with short-selling constraints, and thus the set of attainable liquidation values is closed regarding the convergence in probability. This is shown by Napp~\cite{napp.2003} 
(see Lemma~3.1 and Corollary 4.2 therein).

\begin{Bemerkung}[Wash sales]\label{3.1.2018.2}
From a theoretical perspective, it is interesting to observe that (NA) in the sense of Definition~\ref{22.11.2017.1} does not follow from (NA) in the same model but under the prohibition of wash sales. Indeed, in the spirit of Example~\ref{9.11.2017.4}, one can construct an example with $|\Omega|=4$ 
in which the long-term investment and two one-period  investments in two subsequent periods (with all $4$ combinations of good/bad one-period returns) can be combined to an arbitrage. 
In the first of these periods, the bad return has to become negative 
(and not only smaller than $r$). Through enlarging of the time horizon, this can be achieved without violating the other properties. Subsequently, the arbitrage may disappear if a wash sell between the two periods is forbidden, i.e., one cannot realize the loss in the first period and invest in the second period. We leave the construction of the example as an exercise for the interested reader. 

Nevertheless, for the arbitrage theory with a single non-shortable risky stock, the effect
of the prohibition of wash sales is marginal. This can be seen in Example~\ref{9.11.2017.2},
in which the arbitrage strategy that exists under the larger lower bound for the return in period~$t$ does not make use of wash sales. 
\end{Bemerkung}

\section{Relation to models with proportional transaction costs}\label{29.12.2017.4}

It is interesting to note that the model of Dybvig and Koo~\cite{dyb1}
can be written as a model with proportional transaction costs by introducing several fictitious securities; namely, for every $i=0,\ldots,T-1$, we consider a security~$i$ that can be bought at
the ask-price process~$\ov{S}^i$ and sold at the bid-price process~$\un{S}^i$. In fact, security~$i$ models the purchase of the original stock at time~$i$ and a later liquidation. This means that we put
\beam\label{19.5.2017.6}
\ov{S}^i_i := S_i,\quad \un{S}^i_t := S_t - \alpha(S_t-S_i),\  t=i+1,\ldots,T,
\quad \mbox{and}\ \ov{S}^i=\infty,\ \un{S}^i=-\infty\ \mbox{otherwise.}
\eeam
Note that a short position in security~$i$ cannot be closed after time~$i$.
This forces the investor to hold only nonnegative numbers of the securities.

\begin{Proposition}\label{6.11.2017.2}
Let the assumptions of Proposition~\ref{19.5.2017.1} be satisfied.

(i) The set of attainable terminal wealths $\mathcal{A}:=\{\eta_T\in L^0(\Omega,\mathcal{F},P)\ |\ (\eta,N)\ \mbox{self-financing}\} - L^0_+(\Omega,\mathcal{F},P)$ is closed regarding the convergence in probability.

(ii) There exists a $Q\sim P$ with bounded $dQ/dP$ such that $E_Q(S_i)<\infty$ for all $i=0,1,\ldots,T$ and 
\beam\label{19.5.2017.4}
E_Q\left(\frac{S_{\tau} - \alpha(S_{\tau}-S_i)}{(1+(1-\alpha)r)^{\tau}}
\right) 
\le E_Q\left(\frac{S_i}{(1+(1-\alpha)r)^i}\right)\quad \forall i=0,\ldots,T-1,\ \tau\in\mathcal{T}_i,
\eeam
where $\mathcal{T}_i$ denotes the set of $\{i,\ldots,T\}$-valued stopping times.
\end{Proposition}

\begin{Proposition}\label{15.12.2017.2}
Let $|\Omega|<\infty$. The model satisfies (NA) iff there exists a $Q\sim P$ such that
(\ref{19.5.2017.4}) holds.
\end{Proposition}

\begin{proof}[Proof of Proposition~\ref{6.11.2017.2}]

{\em Ad (i):} The tax model can be identified with the transaction costs model~(\ref{19.5.2017.6}). Let us show that it satisfies the ``robust no-arbitrage'' property for models with transaction costs, as introduced in Definition~1.9 of Schachermayer~\cite{schachermayer.2004}. In the special case that we consider here, this means that if prices are finite, one has to find slightly more favorable prices under which the model still satisfies (NA). Define
\beao
\ov{C}^i_i := S_i\left(1-\frac13\frac{\alpha(1-\alpha)r}{(1+(1-\alpha)r)^{T-i}}\right),\quad
\un{C}^i_t := S_t - \alpha (S_t-S_i) +  \frac13\alpha(1-\alpha)r S_i,
 \eeao 
$i=0,\ldots,T-1,\ t= i+1,\ldots,T$
and again $\un{C}=-\infty,\ \ov{C}=\infty$ otherwise. 
The no-arbitrage property of $(\un{C},\ov{C})$
follows along the lines of the proof of Proposition~\ref{19.5.2017.1}, especially by an inspection
of estimation~(\ref{18.5.2017.3}).
Then, by Theorem~2.1 of \cite{schachermayer.2004}, it follows that $\mathcal{A}$
is closed regarding the convergence in probability.

{\em Ad (ii):} With (i), the proof of existence of a separating measure is 
straightforward. For the convenience of the reader, we briefly repeat the well-known arguments (cf. Schachermayer~\cite{schachermayer.1992} for an easily-accessible overview). First, we define the probability measure~$\wt{P}\sim P$
by $d\wt{P}/dP = c/(1+S_0 +\ldots+S_T)$ and $c:=1/E_P\left(1/(1+S_0 +\ldots+S_T)\right)$, for which one has 
\beam\label{8.1.2018.1}
S_i\in L^1(\Omega,\mathcal{F},\wt{P}),\quad i=0,\ldots,T.
\eeam
(i) and Yan's theorem (see, e.g., Theorem~3.1 in \cite{schachermayer.2004}) imply the existence of a probability measure~$Q$ with positive and bounded $dQ/d\wt{P}$ such that
\beam\label{16.12.2017.1}
E_Q(\eta)\le 0\quad \forall \eta\in\mathcal{A}\cap L^1(\Omega,\mathcal{F},\wt{P}).
\eeam
Consider $(X_{s,u})_{s<u}$ from (\ref{28.11.2017.2}). For any self-financing strategy~$(\eta,N)$, the liquidation value~$\eta_T$ is given by (\ref{27.11.2017.1}).
For $i\in\{0,\ldots,T-1\},\ \tau\in\mathcal{T}_i$, we define $N_{i,t}:=1_{\{\tau>t\}}$ and $N_{s,t}=0$ for $s\not=i$ that yields $\eta_T=V(N)=X_{i,\tau}$ for the
corresponding self-financing $(\eta,N)$, with the convention $X_{i,i}:=0$. (\ref{8.1.2018.1}) implies $X_{i,\tau}\in L^1(\Omega,\mathcal{F},\wt{P})$. By (\ref{16.12.2017.1}), it follows that $E_Q(X_{i,\tau})\le 0$ for all $i=0,\ldots,T$,\ $\tau\in\mathcal{T}_i$. Dividing $X_{i,\tau}$ by
the constant~$(1+(1-\alpha)r)^T$ yields (\ref{19.5.2017.4}). At this point, it is crucial that the interest rate is deterministic. In addition, note that $dQ/dP=dQ/d\wt{P} \cdot d\wt{P}/dP$ is bounded.
\end{proof}

\begin{proof}[Proof of Proposition~\ref{15.12.2017.2}]
On a finite probability space, (NA) is equivalent to the existence of a probability measure~$Q\sim P$
with $E_Q(\eta_T)\le 0$ for all $\eta_T$ from (\ref{27.11.2017.1}), see, e.g., again
Theorem~3.1 in \cite{schachermayer.2004} for the non-trivial direction. As argued above, the latter implies that
\beam\label{15.12.2017.3}
E_Q(X_{i,\tau})\le 0\quad \forall i=0,\ldots,T-1,\ \tau\in\mathcal{T}_i,
\eeam
and it remains to show equivalence.
(\ref{15.12.2017.3}) yields the existence of $Q$-martingales $(M^i_t)_{t=i,\ldots,T}$ with 
\beam\label{19.5.2017.5}
M^i_i=0\quad\mbox{and}\quad M^i_t \ge X_{i,t},\quad t=i+1,\ldots,T.
\eeam
Indeed, let $M^i$ be the martingale part of the Snell-envelope of the process~$(X_{i,t})_{t=i,\ldots,T}$. With representation~(\ref{27.11.2017.1}) and (\ref{19.5.2017.5}), it follows that
\beao
\eta_T & = & \sum_{i=0}^{T-1}\sum_{t=i+1}^T (N_{i,t-1}-N_{i,t})X_{i,t}\\
& \le & \sum_{i=0}^{T-1}\sum_{t=i+1}^T (N_{i,t-1}-N_{i,t})
M_t^i\\
& = & \sum_{i=0}^{T-1}\sum_{t=i+1}^T N_{i,t-1}M_t^i
-\sum_{i=0}^{T-1}\sum_{t=i+2}^{T+1} N_{i,t-1}M_{t-1}^i\\
& = & \sum_{i=0}^{T-1}\sum_{t=i+2}^T N_{i,t-1}
(M_t^i-M^i_{t-1}) + \sum_{i=0}^{T-1}N_{i,i}M^i_{i+1},
\eeao
where for the last equality we use that $N_{i,T}=0$. Since $N_{i,t-1}$ is 
$\mathcal{F}_{t-1}$-measurable and $M^i$ is a $Q$-martingale with $M^i_i=0$, it follows that $E_Q(\eta_T)\le 0$.
\end{proof}

\section{Conclusion}

In models with taxes, no-one-period-arbitrage is only a necessary but
not a sufficient condition for dynamic no-arbitrage. 
Thus, we introduce the robust local no-arbitrage~(RLNA) condition as the weakest local condition on stochastic stock price returns  which guarantees dynamic no-arbitrage. (RLNA) can be verified under a similar dichotomy condition as no-one-period arbitrage (see (\ref{29.8.2017.1}) vs. (\ref{30.10.2017.2})). By comparing the boundary~$\kappa_{t,T}$ under which
the stochastic return in period~$t$ has to fall with positive possibility to exclude arbitrage with the riskless interest rate~$r$ one can estimate how non-local the no-arbitrage property is.
The difference between $\kappa_{t,T}$ and $r$ is quite remarkable since for $T-t\to\infty$, $\kappa_{t,T}$ tends to $-\alpha + (1-\alpha)^2r$,
i.e., for $r\ll \alpha$, the just tolerable potential loss in the stock coincides with the previous stock price times the tax rate. This refers to the extreme case in which 
the stock's purchasing price is negligible compared to its current price
and in addition the investor can  defer accrued taxes forevermore. Thus, she 
tolerates possible losses that are only slightly smaller than the taxes that she has to pay
if the stock is liquidated.  Example~\ref{9.11.2017.2} explains the difference between $\kappa_{t,T}$ and $r$ by the possibility that 
the current stock return can be used as a hedging instrument against the future stock return. The puzzling phenomenon that two long positions in the same stock can hedge each other also has consequences for the existence of a separating measure in an arbitrage-free market. 
The phenomenon cannot occur in arbitrage-free frictionless markets or markets with proportional transactions costs. We show that in the tax model with a bank account and one risky stock, no-arbitrage alone does not imply the existence 
of an equivalent separating probability measure (see Example~\ref{3.1.2018.1}). This is in contrast to models with proportional transaction costs for which Grigoriev's theorem shows the opposite. Furthermore, as a by-product of our analysis, one obtains an example  showing that Grigoriev's theorem cannot be extended to dimension~$3$ (see Example~\ref{26.1.2018.2}). 

\addcontentsline{toc}{section}{References}  

\bibliography{references}

\end{document}